\documentclass[nofootinbib,english, aip, jcp, priprint, graphicx,floatfix]{revtex4-1}
\usepackage{amsmath,bm}
\usepackage{amsthm}
\usepackage{amssymb}
\usepackage{mathrsfs}
\usepackage{bbm,latexsym}
\usepackage{graphicx}
\usepackage{cancel}
\usepackage{enumitem}
\usepackage{wrapfig}
\usepackage{setspace}
\usepackage[bottom]{footmisc}

\newcommand{\tmop}[1]{\ensuremath{\operatorname{#1}}}

\newtheorem{definition}{Definition}
\newtheorem{proposition}{Proposition}
\newtheorem*{proposition*}{Proposition}
\newtheorem*{corollary*}{Corollary}

\newtheorem{theorem}{Theorem}
\newtheorem{lemma}{Lemma}
\newtheorem{algorithm}{Algorithm}

\draft 

\theoremstyle{plain}
\newtheorem{thm}{\protect\theoremname}
\newtheorem*{thm*}{\protect\theoremname}
\newtheorem*{lem*}{\protect\lemmaname}
\theoremstyle{definition}

\theoremstyle{plain}
\newtheorem{cor}[thm]{\protect\corollaryname}

\newcommand{\normal}{{\mathfrak{n}}}

\newcommand{\indicatorf}[1]{\mathbb{I}_{#1}}
\newcommand{\capac}[2]{\ensuremath{\operatorname{cap}}(#1,#2)}
\newcommand{\hausdorffmeasure}{\mathscr{H}(dx)}
\newcommand{\PMeasure}{\mathscr{P}(dx)}
\newcommand{\tPMeasure}{\tilde{\mathscr P}(dx)}

\newcommand{\bb}[1]{\mathcal{B}\left(#1\right)}

\usepackage{babel}
\providecommand{\corollaryname}{Corollary}

\providecommand{\definitionname}{Definition}
\providecommand{\lemmaname}{Lemma}
\providecommand{\theoremname}{Theorem}

\newcommand{\dA}{{\dot A}}
\newcommand{\tA}{{\tilde A}}
\newcommand{\dB}{{\dot B}}
\newcommand{\tB}{{\tilde B}}
\newcommand{\capA}{\kappa_A}
\newcommand{\capB}{\kappa_B}

\usepackage[compact]{titlesec}
\titlespacing{\section}{0pt}{20pt}{20pt}
\titlespacing{\subsection}{0pt}{*0}{*0}
\titlespacing{\subsubsection}{0pt}{*0}{*0}

\begin{document}

\title{Capacities and the Free Passage of Entropic Barriers} 

\author{Jackson Loper}
\thanks{These two authors contributed equally}
\affiliation{Data Science Institute, Columbia University, New York, NY, USA}

\author{Guangyao Zhou}
\thanks{These two authors contributed equally}

\author{Stuart Geman}

\affiliation{Division of Applied Mathematics, Brown University, Providence, RI, USA}
\date{\today}

\begin{abstract}
	We propose an approach for estimating the probability that a given small target, among many, will be the first to be reached in a molecular dynamics simulation. Reaching small targets out of a vast number of possible configurations constitutes an \emph{entropic barrier}. Experimental evidence suggests that entropic barriers are ubiquitous in biomolecular systems, and often characterize the rate-limiting step of biomolecular processes. Presumably for the same reasons, they often characterize the rate-limiting step in simulations. To the extent that first-passage probabilities can be computed without requiring direct simulation, the process of traversing entropic barriers can replaced by a single choice from the computed (``first-passage'') distribution.
We will show that in the presence of certain entropic barriers, first-passage probabilities are approximately invariant to the initial configuration, provided that it is modestly far away from each of the targets.  We will further show that as a consequence of this 
invariance, the first-passage distribution can be well-approximated in terms of ``capacities" of local sets around the targets.  Using these theoretical results and a Monte Carlo mechanism for
approximating capacities, we provide a method for estimating the  hitting probabilities of small targets in the presence of entropic barriers. In numerical experiments with an
idealized (``golf-course'') potential, the estimates are as accurate as the results of direct simulations, but far faster to compute.
\end{abstract}


\maketitle 

\section{Introduction}
\label{sec:Introduction}

Molecular dynamics simulations help us understand a diverse range of biomolecular processes, including the folding of macromolecules into their native configurations\cite{Scheraga2007-qw} and the conformational changes involved in their functioning.\cite{Hospital2015-ol} Since their first introduction in the 1970s,\cite{McCammon1977-kg, Warshel1976-qg} substantial increase in speed and accuracy of molecular dynamics simulations has been achieved. However, we are still severely limited in the timescale we can access. Even with specialized hardwares, we can only achieve atomic-level simulations on timescales as long as milliseconds,\cite{Dror2012-ws} while the timescales for various biomolecular processes vary widely, and can last for seconds or longer.\cite{Naganathan2005-ki, Zemora2010-lb} 

Efforts have been made to extend the timescale accessible by molecular dynamics simulations from two distinct perspectives: A \emph{kinetic perspective} seeks to find ways to directly understand the dynamics and simulate them with less computational effort, usually under the framework of kinetic transition networks\cite{Noe2006-cs, Wales2006-ur} or Markov State Models\cite{Pande2010-yi, Chodera2014-bh, Husic2018-xp}. The reaction rates are given by various methods that build on top of transition state theory\cite{Eyring1935-ur, Chandler1978-bq, Wigner1997-kk}, including transition path sampling\cite{Dellago1998-lb, Bolhuis2002-ws}, transition interface sampling\cite{Van_Erp2005-vw}, and transition path theory\cite{E2006-fm, E2010-sr}. By contrast, a \emph{thermodynamic perspective} aims to understand and efficiently sample from the ``equilibrium distribution'' on the configuration space; simulated annealing,\cite{Kirkpatrick1983-su} genetic algorithms\cite{Goldberg1989-ko} and parallel tempering\cite{Sugita1999-vh} are representative examples of this approach. However, the dynamics are lost, and additional work\cite{Yang2007-gn, Andrec2005-fh, Zheng2009-ow, Huang2010-uu} is needed in order to recover the kinetic information. Besides, not everyone agrees that a true equilibrium (as opposed to something metastable) is always actually reached by real physical systems\cite{Levinthal1968-ov, Baker1994-px}.

In this paper we focus on \emph{hitting probabilities}: Given an initial condition $X_0$ and two targets, $A,B$, what is the probability that we will hit target $A$ first? This question lies in-between the two perspectives. Compared with the kinetic perspective, it lacks certain dynamic information: it does not tell us \emph{how long} it takes to hit each target.  Compared with the thermodynamic perspective, it adds dynamic information.
Namely, it informs the study of state-to-state transitions.  It also appears as an important subproblem for other methods, e.g. transition probability estimation in Markov State Models, committor function estimation in transition path theory, et cetera.

We are especially interested in the ``small-target" regime for the hitting probability question---reaching targets representing a tiny fraction of a vast number of possible configurations. The multitude of configurations constitute an \emph{entropic barrier}. A prototypical example (and a convenient metaphor) is a  ``golf-course'' potential.\cite{bicout2000entropic, Baum1986-we, Wille1987-tf}. These potentials feature several small targets in a much larger region. Away from the targets, the potential energy is relatively flat.  Near the targets, the potential energy may be complicated. Regions of the energy landscape with these small-target golf-course potentials are ubiquitous in biomolecular systems. Folding-like dynamics (such as the folding of RNAs and proteins) are one prominent example, where, quoting 
McLeish\cite{McLeish2005-dq}, ``folding rates are controlled by the rate of diffusion of pieces of the open coil in their search for favorable contacts'' and ``the vast majority of the space covered by the energy landscape must actually be flat.'' Indeed, experimental evidence shows that exploration of these regions with golf-course potentials is the rate-limiting step for a variety of processes.\cite{Teschner1987-qs, Jacob1999-bs, Goldberg1999-mv, Plaxco1998-iv}

Direct simulations offer one way to estimate hitting probabilities, but perform poorly in the presence of entropic barriers. 
Significant, if not most, computation is spent traversing large expanses of nearly flat landscape. In general, the computational efficiency of direct simulations will scale inversely with the size of the targets.

Standard acceleration techniques do not always apply in the presence of entropic barriers.  For example, simulated annealing and parallel tempering work poorly in the presence of regions with golf-course potentials.\cite{Baum1986-we, Wille1987-tf, Machta2009-gh} Some authors take a pessimistic view on this subject: ``If these processes are intrinsically slow, i.e. require an extensive sampling of state space'' (which is indeed the case in the presence of entropic barriers), ``not much can be done to speed up their simulation without destroying the dynamics of the system.''\cite{Christen2008-ge}

In this paper, we argue for a new point of view: the long timescales in the presence of entropic barriers can be a blessing rather than a curse.  If the targets are sufficiently small then the system may reach a ``temporary equilibrium'' before entering any of the complex landscapes found in and around them. The exact initial conditions become irrelevant as the hitting probabilities become approximately independent of where the process started.  Moreover, these hitting probabilities may be approximately \emph{invariant} to the energy landscape away from the targets, meaning we can compute the approximately constant hitting probabilities using only local simulations around the targets.

The following idealization will help to illustrate these ideas: The configuration space,
$\Omega$, is the unit ball in $\mathbb{R}^n$
($\Omega = \bb{0,1}$, where $\bb{x, r} \triangleq \{ y : || y - x || < r \}$). The configuration itself, $X=X_t$, is confined to $\Omega$ by a reflecting boundary at $\partial\Omega$, and within $\Omega$ the dynamics are assumed to be well approximated by the first-order (high-viscosity) Langevin equation,
\begin{equation} 
\label{equ:toy_sde}
\mathrm{d} X_t = - \nabla U (X_t) \mathrm{d} t + \mathrm{d} W_t 
\end{equation}
where $W_t$ is an $n$-dimensional Brownian motion and 
$U$ is a potential energy. There are only two targets, 
\begin{align*}
A &\triangleq \bb {x_A, r_A}\\
B &\triangleq \bb {x_B, r_B}
\end{align*}
both completely contained in $\Omega$ and both small, in the sense that $r_A,r_B \ll 1$.
Given an initial condition $X_0\in \Omega \backslash (A\cup B)$, we are interested to know whether the dynamics carry the system into $A$ or $B$ first.  If we let $h_{A,B}(x)$ be the probability that $X_t$ reaches $A$ before reaching $B$ (``first-passage at $A$''), which depends on the starting configuration $X_0=x$, then more specifically we are interested in computing approximations to $h_{A,B}(x)$ that sidestep the entropic barrier in the special case that U is smooth beyond the immediate neighborhoods of $A$ and $B$.

To this end, introduce neighborhoods $\dA=\bb{x_A,r_\dA }$ and $\tA=\bb{x_A,r_\tA }$ such that $r_A<r_\dA<r_\tA$, and
$\dB=\bb{x_B,r_\dB }$ and $\tB=\bb{x_B,r_\tB }$ such that $r_B<r_\dB<r_\tB$, and assume, in the extreme case,  that $\nabla U(x)=0$ for all $x\in\Omega\backslash (\dA\cup\dB)$
(see Figure \ref{fig:ToyModel}). 
Though $U$ may be arbitrarily complicated on $\dA \cup \dB$, we will show that the hitting probabilities $h_{A,B}(x)$ converge uniformly over $x\in \Omega \backslash (\tA\cup \tB)$ to a constant as the ``action regions'' $\dA$ and $\dB$ are made smaller.
Furthermore, the constant itself depends only on the ratio of
two so-called capacities, $\capac{A}{\tA}$ and $\capac{B}{\tB}$, which are well-known integrals, depending only on $U$, over the sets $\tA\backslash A$ and $\tB\backslash B$.
At this point, then, the problem comes down to devising an efficient scheme for estimating capacities. 
In this regard, the probabilistic interpretation of capacities can be helpful. By exploiting both points of view, the explicit integral representation of capacities and their connections to first-passage probabilities, we will provide a general Monte Carlo approach to their estimation. 
Of course there is nothing special about two targets, in any of our results. Even if there are more, the first-passage probabilities are still determined by the ratios of the target capacities. 

\begin{wrapfigure}
[18]{r}[0pt]{0.5\textwidth}
	\centering	\includegraphics[width=0.4\textwidth]{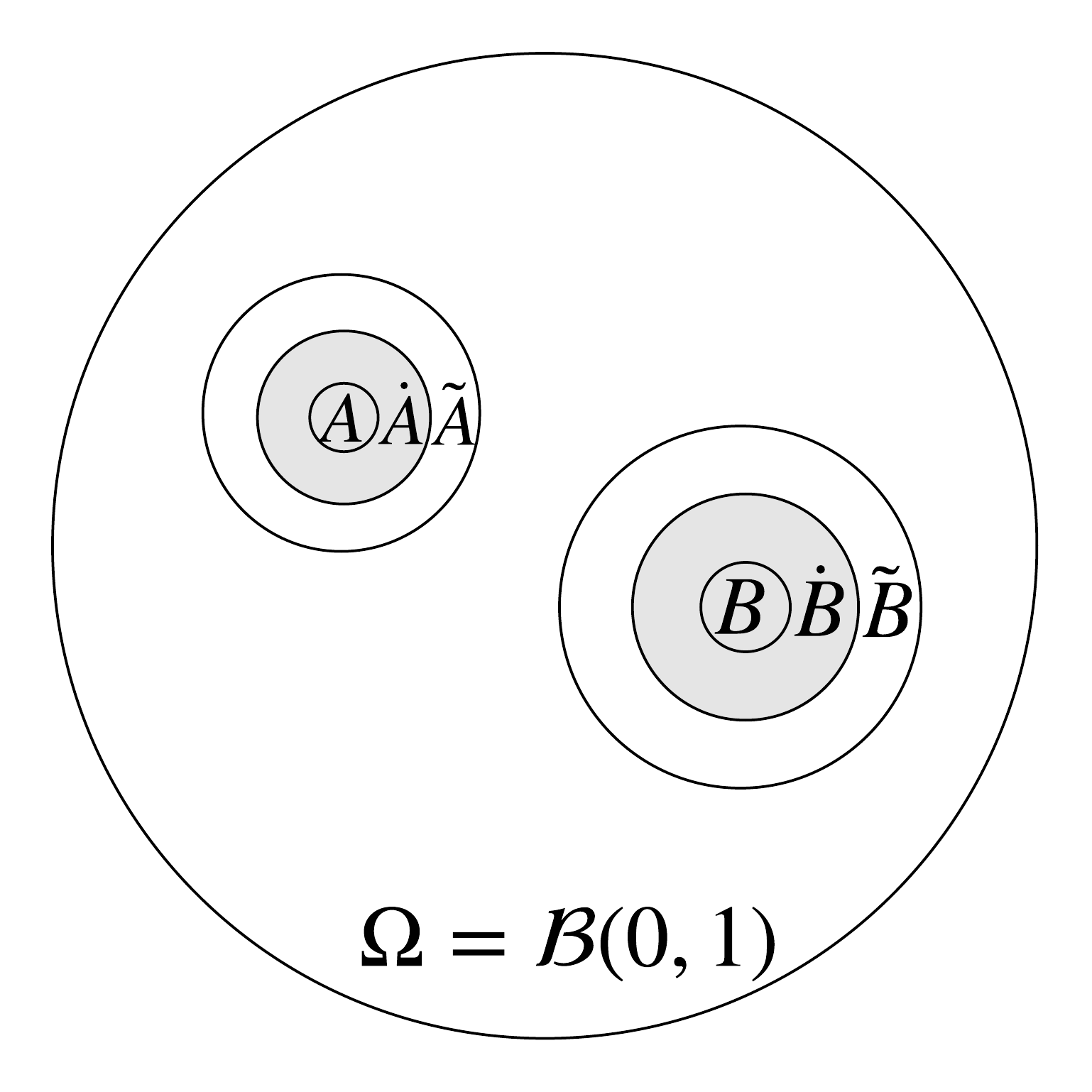}
	\caption{\footnotesize\linespread{1.}\selectfont{} {\bf A Toy Model.} The dynamics obey the first-order Langevin equation (\ref{equ:toy_sde}). There are two targets, $A$ and $B$, each surrounded by two neighborhoods. Outside of $\dA\cup\dB$, the energy is flat ($\nabla U=0$). If the targets are small enough or the dimension large enough then the probability of entering $A$ before entering $B$  is nearly independent of $X_0$, provided that $X_0$ is outside of $\tA\cup\tB$. What is more, this first-passage probability depends only the behavior of $U$ in $\dA\cup\dB$, which may be arbitrarily complicated.}
\label{fig:ToyModel}
\end{wrapfigure}

We experimented with the approximate Langevin equation (\ref{equ:toy_sde}) in $n=5$ dimensions with two targets. 
We found excellent agreement between direct simulations for computing $h_{A,B}(x)$, on the one hand, and an application of our theoretical and algorithmic tools, on the other hand. Furthermore, in that our approach to estimating capacities is based on a small number of sequential steps, each of which involves thousands of repeated and independent applications of a single stochastic procedure, it admits to almost arbitrary acceleration through parallelization.

Our results are relevant in a small-target regime. We remark that such regimes arise naturally in high-dimensional configuration spaces, since the ratio of the volume of the configuration space to the target volume grows rapidly with increasing dimension.
The effect on dynamics can be readily seen by a simple change of variables to spherical coordinates centered, say, at the center of the target $A$. The result, in terms of the radial coordinate $R=||X-x_A||$, is an additional (radial) drift equal to $\frac{n-1}{2R}dt$. Thus in high dimensions, $n$, the entropic effect amounts to a force that pushes configurations away from small targets. 

In sum, there are three contributions in this paper. The first 
is to give a concrete example in which the probability that a particular small
target, among many, is reached first is essentially independent of the starting location, provided that $U$ is flat outside of a neighborhood of the targets and that the starting position is sufficiently far away. The second, also an approximation result, is to show that a consequence of this (near) independence, whether or not it arises from a flat energy, is that the probability of first hitting a particular target is nearly proportional to a suitably defined local capacity.  
The third is a collection of tools, both analytic and numeric (Monte Carlo), for computing the relevant capacities.

In the following section (\S\ref{sec:Preliminaries}, ``Preliminaries'') we will introduce a more general setting as well as the notation needed to set up the two approximation theorems, which appear in \S\ref{sec:MainResults}.
After that, in \S\ref{sec:Estimation}, an approach to capacity estimation is developed, followed, in  \S\ref{sec:Experiments}, by computational experiments with a five-dimensional version of (\ref{equ:toy_sde}). We will close in \S\ref{sec:Discussion} with a summary, and a discussion of the prospects for 
applying these approximations, recursively, to help illuminate folding pathways. There are many challenges. We will attempt to 
highlight the most important of these.

\section{Preliminaries}
\label{sec:Preliminaries}

In the most general setting, our results are about a diffusion process $X_t\in\mathbbm{R}^n$ confined to a bounded open set $\Omega \subset \mathbb{R}^n$ with reflecting smooth boundary $\partial\Omega$: 
\begin{equation}\label{equ:general_sde}\mathrm{d} X_t = b (X_t) \mathrm{d} t + \sigma (X_t) \mathrm{d} W_t \end{equation}
where $W_t$ is an $n$-dimensional standard Brownian motion, and $b: \Omega \rightarrow \mathbbm{R}^n$ and $\sigma :
\Omega \rightarrow \mathbbm{R}^{n \times n}$ are continuously differentiable vector-valued and matrix-valued functions. 

Let $\bar \Omega$ denote the closure of $\Omega$ (and in general let $\bar S$ denote the closure of any set $S\subset \bar \Omega$).  For the precise definition of the reflected process, we adopt the framework developed by Lions and Sznitman\cite{lions1984stochastic}: Let $\normal=\normal(x)$ denote the outward normal of $\partial \Omega$ and $\nu:\ \partial \Omega \rightarrow \mathbb{R}^n$ a smooth vector field satisfying $\normal^T\nu\geq c>0$, and assume that $x_0 \in \Omega$.  Then there is a unique pathwise continuous
and $W$-adapted strong Markov process $X_t\in\bar\Omega$, and (random) measure $L$, such that 
\begin{gather}\label{eq:SDER}
X_t = x_0 + \int_0^t b(X_s)ds + \int_0^t \sigma(X_s)dW_s - \int_0^t \nu(X_s) L(ds)
\end{gather}
and $L(\{t:\ X_t \notin \partial \Omega\})=0$. 
For convenience, we will refer to $X$ by simply saying ``the reflected diffusion process (\ref{equ:general_sde}).'' 

Assume that $U:\bar \Omega \rightarrow \mathbb{R}$ is continuously differentiable. Our interest is in a reversible process with equilibrium  
\[
\rho(dx)\doteq \frac{1}{Z}e^{-U(x)}\ \ \
Z=\int_{x\in\Omega}e^{-U(x)}dx
\]
for which, as shown by Chen\cite{chen1993reflecting}, it is sufficient that 
\begin{align}
\begin{split}
b_i(x)&=\frac{1}{2} \sum_j \partial a_{ij}(x)/\partial x_j - \frac{1}{2}\sum_j a_{ij}(x) \partial U(x)/\partial x_j 
\label{eqn:reversibility} \\
\nu(x)&= a(x) \normal(x) 
\end{split}
\end{align}
where $a(x)=\sigma(x)\sigma(x)^T$ is uniformly elliptic.
When the conditions in (\ref{eqn:reversibility}) are in force
we will say that $X$ satisfies the reversibility conditions relative to $U$.

Given a region $S$ in $\Omega$,  define $\tau_S \triangleq \inf \{ t \geqslant 0 : X_t \in S \}$, i.e. the time when $X$ first hits $S$. Given two targets, $A,B\subset \Omega$, and an initial condition $x\in\Omega\backslash(A\cup B)$, define the $h_{A,B}(x)$ to
be the probability that $X_t$, starting at $x$, visits $A$ before $B$:
\[ h_{A, B}(x) \triangleq \mathbb{P}(X_{\tau_{A\cup B}}\in A|X_0=x)\]
And finally, we will say that a function $f(x)$ defined on a set
${\cal M}$ is 
``$\varepsilon$-flat relative to $M$'' if $M\subset {\cal M}$ and
\[
\sup_{x, y \in M} |f(x) - f(y)| < \varepsilon
\]

\section{Main theoretical results}
\label{sec:MainResults}

The key to going from an assumed smooth landscape outside of the immediate neighborhoods of targets to the conclusion that first-passage probabilities depend only on certain local capacities is the intermediate conclusion that in the small-target (or large-dimension) limit, those probabilities are nearly independent of starting location, i.e. $\varepsilon$-flat relative to all of $\Omega$, except some small neighborhoods around the targets.
Sufficient conditions for this $\varepsilon$-flatness is the purpose of our first result, and the assumption of 
$\varepsilon$-flatness is the main hypothesis of our second result, which then provides an explicit formula for the first-passage probabilities in terms of certain local capacities.

The intent of the first result is two-fold: give a concrete setting in which the $\varepsilon$-flatness of first-passage probabilities can be rigorously proven, and, in the proof, establish a possible road-map for approaching the problem in other, problem-specific, settings. For these purposes, we adopt the setup laid out in \S\ref{sec:Introduction}, Equation (\ref{equ:toy_sde}) (``toy problem'').
The second result is proven under far more general conditions.

\subsection{Approximately Constant First-passage Probabilities}
\label{subsec:ApproximatelyConstant}

Return to the model introduced in Equation (\ref{equ:toy_sde}).
This toy model encapsulates the essential characteristics of an entropic barrier: if $r_A, r_B$ are small or the dimension $n$ is high, starting from an initial condition $X_0 \not\in \bb{x_A,r_\dA} \cup \bb{x_B,r_{\dot B}}$, the system faces the difficulty of reaching the small targets $A$ and $B$ out of all other configurations in $\Omega$. The nontrivial energy landscapes in the immediate vicinity of $A$ and $B$ reflect the energetic interactions that are usually local in nature in biomolecular systems, and the reflecting boundary at $\partial\Omega$ captures the notion that not all configurations in biomolecular systems are sensible, because of, for example, limits on bond lengths, angles and dihedral angles in the case of RNA molecules.

Assume that $U$ is continuously differentiable.  Recall that $\Omega=\bb{0,1}$ and the targets 
$A=\bb{x_A,r_A }$ and $B=\bb{x_B,r_B }$ are surrounded by two levels of neighborhoods, 
$\dA=\bb{x_A,r_\dA }$ and $\tA=\bb{x_A,r_\tA }$ such that $A \subset \dA \subset \tA$, and 
$\dB=\bb{x_B,r_\dB }$ and $\tB=\bb{x_B,r_\tB }$ such that
$B \subset \dB \subset \tB$. Assume that $\tA$ and $\tB$ are disjoint, and, for convenience, that $\tA$ and $\tB$ are contained entirely within $\Omega$.
\begin{theorem}\label{thm:epsilon_flat}
If $\nabla U(x)=0$ on $x\in\Omega\backslash(\tA\cup\tB)$, then
for any fixed value of the dimension $n \geq 3$ and any $r_{\tilde{A}}, r_{\tilde{B}}, \varepsilon > 0$, there exists a constant $c=c(n, r_{\tilde{A}}, r_{\tilde{B}}, \varepsilon)$ such that if $r_{\dA}, r_{\dB} < c$ then 
$h_{A,B}(x)$ is 
$\varepsilon$-flat  relative to 
$\Omega / (\tilde{A} \cup \tilde{B})$.  
Likewise, for any fixed values of $r_{\dA}, r_{\tilde{A}}, r_{\dB}, r_{\tilde{B}}, \varepsilon>0$, there exists a constant $c=c(r_\dA, r_{\tilde{A}}, r_\dB, r_{\tilde{B}}, \varepsilon)$ such that if $n \geq c$ then 
$h_{A,B}(x)$ is
$\varepsilon$-flat relative to 
$\Omega / (\tilde{A} \cup \tilde{B})$. 
\end{theorem}

The condition $\nabla U=0$ is severe and unrealistic.
However, inspection of the proof, which we defer to 
Appendix \ref{sec:proof_epsilon_flat}, demonstrates that the key for establishing $\varepsilon$-flatness is a proper separation of time scales: it takes a short time for the process $X$ to reach temporary equilibrium in $\Omega / (\tilde{A} \cup \tilde{B})$ and a long time for the process $X$ to hit the targets $A$ and $B$.  Any model with these characteristics (in particular more general entropic barriers) will feature $\varepsilon$-flat hitting probabilities.

\subsection{Hitting Probabilities and Capacities}
\label{subsec:Capacities}

Assume now, more generally, that $X$ is  
the reflected diffusion process (\ref{equ:general_sde}) in a bounded open set $\Omega \subset \mathbb{R}^n$, satisfying the reversibility conditions in (\ref{eqn:reversibility}) relative to a continuously differentiable potential $U$ on $\Omega$.
The concentric spherical neighborhoods of the toy problem are replaced by nested sets 
$A \subset\tA\subset\Omega$, and 
$B \subset\tB\subset\Omega$,
all with smooth boundaries, and where $\tA\cap\tB=\emptyset$.

Under the assumption that $h_{A,B}(x)$ is $\varepsilon$-flat
relative to $\Omega \backslash (\tilde A \cup \tilde B)$, 
we will show that the first-passage probabilities can be accurately estimated using only local information around the targets. In particular, it will be sufficient to calculate the ``capacities'' of 
the sets $\tA \backslash A$ and $\tB \backslash B$:

\begin{definition}(Capacity)
For $S \subset \tilde{S} \subset \Omega$, we define the capacity of $\ensuremath{\operatorname{cap}} (S, \tilde{S})$ relative to $U$ as
\[ \ensuremath{\operatorname{cap}} (S, \tilde{S}) \triangleq \int_{\tilde S \backslash S}
||\sigma(x) \nabla h_{S, \tilde{S}^c}(x)||^2 e^{- U(x)} \mathrm{d} x \]
\end{definition}

We refer the reader to Appendix \ref{sec:three_perspectives} for more details on capacity and the related concept of Dirichlet form.  Note that there are several related definitions of ``capacity" in the literature.  Throughout this work, the term is used only in the sense of the above definition. 

In the presence of $\varepsilon$-flatness, these local capacities are intimately related to global hitting probabilities:

\begin{theorem}\label{thm:main_thm}  
Assume that  $h_{A,B}(x)$ is $\varepsilon$-flat relative to 
$\Omega \backslash (\tilde A \cup \tilde B)$.
Then the first-passage probabilities can be well-approximated in terms of the target capacities:
\[ \sup_{x \notin \tilde A,\tilde B} \left| h_{A,B} (x) - \frac{\capac{A}{\tilde A}}{\capac{A}{\tilde A}+\capac{B}{\tilde B}} \right| \leqslant \varepsilon + \sqrt{\varepsilon/2} \]
\end{theorem}

We defer the proof to Appendix \ref{sec:proof_thm}, but here we note that, on account of the additive property of 
capacities (Proposition \ref{prop:capacity} in Appendix \ref{sec:three_perspectives}), the generalization to multiple targets is straightforward.
Assume that $A_1\cdots A_m \subset \Omega$ and $\tilde A_1\cdots \tilde A_m \subset \Omega$, such that $A_k \subset \tilde A_k$ for every $k=1,2,\ldots,m$ and $\tA_1,\ldots\tA_m$ are disjoint. Define 
\begin{equation*}
p_{A_k} = \frac{\ensuremath{\operatorname{cap}} (A_k, \tilde{A}_k)}{\sum_{i = 1}^m \ensuremath{\operatorname{cap}} (A_i, \tilde{A}_i)}, k=1,\dots, m
\end{equation*} 
and, for each $k$, let $u_k(x)$ be the probability that $X_t$ hits $A_k$ before any other target, given that 
$X_0=x\in M\doteq \Omega\backslash\bigcup_{k = 1}^n \tilde{A}_k $.
Then, given the $\varepsilon$-flatness of the functions $\{u_k\}_{k=1,\dots, n}$ relative to $M$, each $u_k$ is well-approximated by $p_{A_k}$:
\begin{cor}\label{thm:main_cor} 
If, for each $k=1,2,\ldots,m$, $u_k$ is 
$\varepsilon$-flat relative to $M$, then
\[ \sup_{x \in M} \left| u_k (x) - p_{A_k} \right| \leqslant \varepsilon + \sqrt{\varepsilon/2}, k=1,\dots, m\]
\end{cor}

\section{Capacity Estimation} 
\label{sec:Estimation}
For a class of stochastic systems, characterized by a separation of time scales such that the process of finding targets is slow compared to the process of exploring the regions away from the targets, we have reduced the estimation of first-passage probabilities to the evaluation, or approximation, of capacities.
Generically, given the process defined in Equation (\ref{equ:general_sde}), satisfying the reversibility conditions in (\ref{eqn:reversibility}) relative to a continuously differentiable energy $U$, our goal is to evaluate
\begin{equation}
\label{eqn:capacity}
\ensuremath{\operatorname{cap}} (A, \tilde{A}) = \int_{\tilde A \backslash A}
||\sigma(x) \nabla h_{A, \tilde{A}^c}(x)||^2 e^{- U(x)} \mathrm{d} x 
\end{equation}
for a target $A$ and neighborhood $\tA$.

The calculation is local, in that $\capac{A,\tilde{A}}$ depends only on the behavior of  $U$ on $\tA\backslash A$, but it is not uncomplicated. We will propose here a Monte Carlo approach to evaluating the integral, made up of a combination of analytic reductions and highly orchestrated random walks. Inevitably, the effectiveness, or even feasibility, of the approach will depend on the particulars of the stochastic system, (\ref{equ:general_sde}).

We begin by replacing the volume integral in (\ref{eqn:capacity}) with a surface integral:

\begin{proposition}
\label{prop:flux}
For any regions $G$ and $\tilde{G}$ having smooth boundaries and such that $A\subset G \subset \tilde G \subset \tilde A$, $\capac{A,\tA}$ can be expressed as a flux leaving $\tilde G \backslash G$:
\begin{equation}
\label{eqn:GIntegral}
\ensuremath{\operatorname{cap}} (A, \tilde{A}) = \int_{\partial (\tilde G \backslash G)}  h_{A, \tilde{A}^c} (x)   \normal(x)^T a (x) \nabla h_{G, \tilde{G}^c} (x)e^{- U (x)} \hausdorffmeasure
\end{equation}
where $a(x)=\sigma(x)\sigma(x)^T$ is the diffusion matrix, $\hausdorffmeasure$ is the $(n-1)$-dimensional Hausdorff measure, and $\normal$ represents the outward-facing (relative to the set $\tilde G \backslash G$) normal vector on $\partial (\tilde G \backslash G)$.
\end{proposition}
\noindent The proof is in Appendix \ref{sec:proof_proposition}.

There is a great deal of freedom in choosing $G$ and $\tilde G$; the idea is to choose them so as to make the surface integrals as simple as possible. Before pursuing this, we mention that there are many other ways to reduce the volume integral (\ref{eqn:capacity}) to a flux integral, some of which might make more sense than (\ref{eqn:GIntegral}) for a particular problem. Specifically, by a corollary of 
Proposition (\ref{prop:flux}), $\capac{A}{\tilde A}$ can be written as the flux of a different field, but this time through a single surface (see Appendix \ref{sec:proof_proposition}):
\begin{cor}
\label{cor:flux}
For any region $S$ having smooth boundary $\partial S$, and such that $A\subset S \subset \tilde A$, $\capac{A,\tA}$ can be expressed as a flux leaving $S$:
\begin{equation}
\ensuremath{\operatorname{cap}} (A, \tilde{A}) = \int_{\partial S}   \normal(x)^T a (x) \nabla h_{A, \tilde{A}^c} (x)e^{- U (x)} \hausdorffmeasure
\end{equation}
where $a$ and $\hausdorffmeasure$ are as defined in the Proposition, and  $\normal$ is the outward-facing normal on $\partial S$.
\end{cor}
\noindent
The possible advantage is that there is only one surface and the integrand involves only one first-passage probability function, $h_{A, \tilde{A}^c} (x)$, instead of two. The possible disadvantage is the need to estimate $\nabla h_{A, \tilde{A}^c}$ on $S$, which is harder than estimating $h_{A, \tilde{A}^c}$. As we will see shortly, judicious choices for $G$ and $\tilde G$ can mitigate, and in some cases even eliminate, the need to estimate gradients of first-passage probabilities.

Returning to the representation in (\ref{eqn:GIntegral}), there are two surface integrals, each of which can be viewed as an expectation, as follows: Define a probability measure on
$\partial G$ by
\[
\PMeasure\doteq\frac{1}{Z}e^{-U(x)}\hausdorffmeasure 
\text{   where   }
Z= \int_{\partial G} e^{-U(x)}\hausdorffmeasure
\]
and define $\tPMeasure$
and $\tilde{Z}$ analogously, but on $\partial\tilde{G}$ rather than $\partial G$.
Then
\begin{align}
\capac{A}{\tilde A} & = \int_{\partial\tilde{G}} h_{A, \tilde{A}^c}    \normal^T a  \nabla h_{G, \tilde{G}^c} e^{- U } \hausdorffmeasure
-\int_{\partial G} h_{A, \tilde{A}^c}    \normal^T a  \nabla h_{G, \tilde{G}^c} e^{- U } \hausdorffmeasure
\nonumber \\
&=Z
\int_{\partial\tilde{G}} h_{A, \tilde{A}^c}    \normal^T a  \nabla h_{G, \tilde{G}^c} \tPMeasure
-Z
\int_{\partial G} h_{A, \tilde{A}^c}    \normal^T a  
\nabla h_{G,\tilde{G}^c} \PMeasure
\label{eqn:PMeasureInt}
\end{align}
where, in these integrals, the normal, $\normal$, points outward from {\em both} $G$ and $\tilde G$.
If now $y_1,y_2,\dots,y_m\sim \text{iid}\ \PMeasure$, then 
\begin{align*}
\frac{1}{m}\sum_{i=1}^m 
h_{A, \tilde{A}^c}(y_i) \normal^T(y_i) a(y_i)  \nabla h_{G,\tilde{G}^c}(y_i) & \stackrel{m\to\infty}{\longrightarrow}
 \int_{\partial G} h_{A, \tilde{A}^c}    \normal^T a  
\nabla h_{G,\tilde{G}^c} \PMeasure \\
\text{and}\ \ \ \ \ \ \frac{1}{m}\sum_{i=1}^m e^{U(y_i)}
& \stackrel{m\to\infty}{\longrightarrow} \int_{\partial G} e^U
\PMeasure = \frac{|\partial G|}{Z}
\end{align*}
where $|\partial G|$ is the surface area of $G$. Putting these together, we get the large $n$ approximation
\[
Z
\int_{\partial G} h_{A, \tilde{A}^c}    \normal^T a  
\nabla h_{G,\tilde{G}^c} \PMeasure
\approx
|\partial G|
\frac{\sum_{i=1}^m 
h_{A, \tilde{A}^c}(y_i) \normal^T(y_i) a(y_i)  \nabla h_{G,\tilde{G}^c}(y_i)}
{\sum_{i=1}^m e^{U(y_i)}}
\]
If we now extend all of this to $\partial\tilde{G}$, with 
$\tilde{y}_1,\tilde{y}_2,\dots,\tilde{y}_n\sim \text{iid}\ \tPMeasure$, and put the approximations into 
(\ref{eqn:PMeasureInt}), then for large $n$ and $m$
\begin{align}
\capac{A}{\tilde A} & \approx
|\partial\tilde{G}|
\frac{\sum_{i=1}^n 
h_{A, \tilde{A}^c}(\tilde{y}_i) \normal^T(\tilde{y}_i) a(\tilde{y}_i)  \nabla h_{G,\tilde{G}^c}(\tilde{y}_i)}
{\sum_{i=1}^n e^{U(\tilde{y}_i)}}
\label{eqn:approximate_capacity}\\
& -|\partial G|
\frac{\sum_{i=1}^m 
h_{A, \tilde{A}^c}(y_i) \normal^T(y_i) a(y_i)  \nabla h_{G,\tilde{G}^c}(y_i)}
{\sum_{i=1}^m e^{U(y_i)}}
\nonumber
\end{align}

To make this useful, we will need to choose $G$ and $\tilde{G}$ so that (i) we can readily sample from $\PMeasure$ and $\tPMeasure$; (ii) the surface areas $|\partial G|$ and
$|\partial\tilde{G}|$ can be well approximated; (iii) 
the first-passage probability $h_{A, \tilde{A}^c}$ can be well approximated on  $G$ and $\tilde{G}$; and (iv) the gradient
$\nabla h_{G, \tilde{G}^c}$ can also be well approximated on 
 $G$ and $\tilde{G}$. 
The first two of these challenges lend themselves to more-or-less routine, though not necessarily easy methods, including importance and rejection sampling. Of course we're free to choose $G$ and $\tilde G$ to make (i) and (ii) as easy as possible.

As for approximating first-passage probabilities and their gradients, broadly speaking there are two approaches. It is well known that first-passage probabilities satisfy an elliptic PDE related to the infinitesimal generator---see Appendix \ref{sec:three_perspectives}, Equation (\ref{eq:pde})---and we could therefore choose from a selection of numerical solvers. One drawback with this approach is that numerical PDE methods are famously difficult to employ successfully in high dimensions (``curse of dimensionality''). Here, in a different direction, we exploit the connection between first-passage probabilities and the underlying random walk in order to develop Monte Carlo tools suitable for estimating both $h_{A, \tilde{A}^c}$ and $\nabla h_{G, \tilde{G}^c}$ on the surfaces
$\partial G$ and $\partial\tilde{G}$. These tools are based on
what we will call the ``shell method,'' which we describe briefly in the following paragraphs and in full detail in Appendix \ref{sec:shell_method}.

Generically, given two simply-connected regions $R$ and $\tilde{R}$, with $R\subset \tilde{R}$, and a set $S$ such that 
$R\subset S  \subset\tilde{R}$, we seek an approximation to 
the function $h_{R,\tilde{R}^c}$ on the surface $\partial S$. In principle, we could begin with a fine-grained partitioning of $\partial S$ into simply-connected ``cells,'' and for each cell run the diffusion $X_t$ many times, recording whether or not the path first exits $\partial (\tilde{R}\backslash R)$ at $\partial R$. The fraction of paths that first exit at $\partial R$ constitutes an estimate of 
$h_{R,\tilde{R}^c}(x)$ for any $x$ in the current cell. But this is wasteful and likely infeasible in all but the simplest of settings. Much of the waste stems from the fact that the ensemble of all paths generated from all cells will likely include many near collisions of paths scattered throughout $\partial (\tilde{R}\backslash R)$. An alternative, divide-and-conquer  approach, is to introduce multiple sets, $S_0,S_1,\ldots,S_n$ such that 
\begin{equation*}
R=S_0 \subset \cdots \subset S_{m-1}\subset S_m = S \subset S_{m+1} \subset \cdots \subset S_n = \tilde{R}
\end{equation*}
and use sample paths from $X$, {\em locally}, to estimate the transition probability matrices from each cell within each ``shell'' $\partial S_k$ to each cell of its neighboring shells, $\partial S_{k-1}$ and $\partial S_{k+1}$. Equipped with these transition matrices, the first-passage probability for a given $x\in S$ is computed algebraically, without further approximation.  

$S$ must have been chosen not only to satisfy $R\subset S  \subset\tilde{R}$ but also in such a way as to make it feasible to sample from $\partial S$ under the probability measure
$\frac{1}{Z}e^{-U}\hausdorffmeasure$. After that, $S_k\ k=1,\ldots,n-1$ are chosen so that the shells nest and are in close proximity; the hitting times starting from a sample in $\partial
 S_k$ and ending at $\partial S_{k-1} \cup \partial S_{k+1}$ must be short enough to encourage many repeated runs. The output is a set of samples, 
 $z_1,\ldots,z_N \sim \frac{1}{Z}e^{-U}\hausdorffmeasure$ on $\partial S$ together with the approximate value of $h_{R,\tilde{R}^c}(x)$ at each sample $x=z_i$. (In fact, though the main purpose is to estimate $h_{R,\tilde{R}^c}$ on 
 $\partial S$, a byproduct is a sample from $\frac{1}{Z}e^{-U}\hausdorffmeasure$ on all of the shells $\partial S_k$, along with an estimate of $h_{R,\tilde{R}^c}$ at every sample.)
With the choice of $A$ for $R$ and $\tilde A$ for $\tilde R$, the algorithm becomes directly applicable to the estimation of $h_{A, \tilde{A}^c}$ on $\partial G$ and $\partial\tilde{G}$, taking $S=G$ in the former case and $S=\tilde{G}$ in the latter.
 
The shell method is closely related to milestoning\cite{West2007-cn, Bello-Rivas2015-ld, Aristoff2016-gc} and Markov state models\cite{Pande2010-yi, Chodera2014-bh, Husic2018-xp}, though more tailored to the problem at hand. In particular, our interest here is in computing the first-passage probabilities rather than in approximating the underlying process. Also, the discretizations of the shells are {\em adaptive}, in that they are based on clusters that are derived from an ensemble of samples, as opposed to being crafted for a particular landscape.  See Appendix \ref{sec:shell_method}.
 
As for the required gradients, these are generally harder to estimate. Nevertheless, for the particular gradient $\nabla h_{G, \tilde{G}^c}$, the problem is substantially mitigated by noting that we are only interested in its evaluation on $\partial G$ and $\partial\tilde G$, each of which is a level set of 
$h_{G, \tilde{G}^c}$ ($h_{G, \tilde{G}^c}=1$ on $G$ and 0 on $\tilde G$). Consequently, on each surface the gradient is in the normal direction and we need only estimate its magnitude. And for this purpose it is enough to know the values of $h_{G, \tilde{G}^c}$ on a surface close to $G$ and interior to $\tilde{G}\backslash G$ (for estimating $\nabla h_{G, \tilde{G}^c}$ on $G$) and on another surface  
close to $\tilde{G}$ and also interior to $\tilde{G}\backslash G$ (for estimating $\nabla h_{G, \tilde{G}^c}$ on $\tilde G$). Two such surfaces would be $\partial S_1$ and $\partial S_{n-1}$, were we to apply the shell method with $R=G$ and $\tilde{R}=\tilde{G}$,
since, as already noted, a byproduct of the method is an estimate of 
$h_{R,\tilde{R}^c}$ on all of the shells. Alternatively, in the interest of better accuracy, the method could be run twice, once with $S=S_1$, a well-chosen outer approximation of $G$, and then again with 
$S=S_{n-1}$, a well-chosen inner approximation of $\tilde G$. 


 \section{Numerical Experiments\footnote{All the experimental results can be reproduced or easily modified from open-source code, which can found, along with detailed instructions, at \url{https://github.com/StannisZhou/entropic_barrier}.}}
\label{sec:Experiments}
We experimented with the two-target system discussed in \S\ref{sec:Introduction} and depicted in Figure \ref{fig:ToyModel}, with $n=5$ dimensions 
and the particular targets $A=\bb{x_A, r_A}$,  where $x_A=(0.5,0.6,0.0,0,0,0.0)$ and $r_A=0.02$, and 
$B=\bb{x_B, r_B}$,  where $x_B=(-0.7,0.0,0.0,0,0,0.0)$ and $r_B=0.04$. The configuration space is the unit ball centered at 
the origin, $\Omega=\bb{0,1}$. 
The goal is to estimate $h_{A,B}(x)$, the probability that $A$ is visited before $B$, using only the behavior of $U$ in the vicinity of the targets, provided that $x$, the starting configuration, is sufficiently far from $A\cup B$. The entropic barrier is idealized by assuming that $\nabla U(x)=0$ outside of $\dA=\bb{x_\dA, r_\dA}$ and $\dB=\bb{x_\dB, r_\dB}$ and ``sufficiently far away'' means outside of $\tA \cup \tB$, where $\tA=\bb{x_\tA, r_\tA}$ and $\tB=\bb{x_\tB, r_\tB}$. The concentric neighborhoods around the targets are supposed to satisfy $A\subset\dA\subset\tA$ and 
$B\subset\dB\subset\tB$, which was enforced in our experiments by the choices $r_\dA= 0.05$, $r_\tA=0.1$, $r_\dB=0.075$, and
$r_\tB=0.15$.
Our experiments test the overall approximation to $h_{A,B}$ developed in \S\ref{sec:MainResults}-\ref{sec:Estimation} as well as each of the three components, separately: $\epsilon$-flatness, the role of capacities, and the methodology developed for estimating capacities.

The experimental setup is sufficiently simple to allow
exhaustive simulation for approximating ground truth. In each experiment, we compare the results of using the approximations developed here to the results from an ensemble of first-passage events simulated by simply running the diffusion 2,000 times at each of 100 randomly chosen points in $\Omega\backslash (\tA,\tB)$. There are two sets of experiments: In the first (a kind of ``sanity check''), the potential $U$ is flat everywhere outside of the targets $A$ and $B$, in other words, the diffusion is Brownian motion. In the second, there are complex landscapes in the vicinities of the targets, i.e. within $\dA$ and $\dB$.

Both Theorems, \ref{thm:epsilon_flat} and \ref{thm:main_thm}, are in force, and hence the first-passage probability $h_{A,B}(x)$, on $\Omega\backslash(\tA,\tB)$, is approximately a constant, and the value depends only on the two local capacities $\capac{A}{\tA}$ and
$\capac{B}{\tB}$:
\[
 h_{A,B} (x) \approx \frac{\capac{A}{\tilde A}}{\capac{A}{\tilde A}+\capac{B}{\tilde B}}
 \]
We need to compute, or approximate, $\capac{A}{\tA}$ and $\capac{B}{\tB}$. Will will work through the details for $\capac{A}{\tA}$, but the identical considerations apply to 
$\capac{B}{\tilde B}$.  We start with the flux representation
established in the Proposition, and the numerical approximation from Equation (\ref{eqn:approximate_capacity}), which reduces the problem to selecting $G$ and $\tilde G$ and then applying the shell method, as described in \S\ref{sec:Estimation}.
In the current setup, 
good choices for $G$ and $\tilde G$ are $G=\dA$ and 
$\tilde G = \tA$, as can be seen from the following observations:
\begin{enumerate}

\item Recall that $h_{A,\tA^c}(x)$ is the probability of first exiting $\tilde{A}\backslash A$ at $\partial A$ rather than at $\partial\tA$, given that the process started at $x$. 
Consequently $h_{A,\tA^c}(x)=0$ for all $x\in\partial\tA$, and hence also on $\partial\tilde{G}$.
Hence, with reference to Equation (\ref{eqn:approximate_capacity}), we need only consider the flux approximation on $\partial G$:
\begin{equation}
\label{eqn:toy_cap_approx}
\capac{A}{\tilde A} \approx
-|\partial G|
\frac{\sum_{i=1}^m 
h_{A, \tilde{A}^c}(y_i) \normal(y_i)\cdot \nabla h_{G,\tilde{G}^c}(y_i)}
{\sum_{i=1}^m e^{U(y_i)}}
\end{equation}
where $y_1,\ldots,y_m$ are independent samples from 
$\PMeasure=\frac{1}{Z}e^{-U}\hausdorffmeasure$ on $\partial G$, $\normal(x)$ faces outward from $G$, and compared to (\ref{eqn:approximate_capacity}), we have used the fact that in the current setup $a(x)$ is the identity $I$.
 
 \item The surface area $|\partial G|$ is just the area of the 4-sphere (in five dimensions), with radius $r_\dA$:
 \[
 |\partial G| = \frac{2\pi^{\frac{5}{2}}}{\Gamma(\frac{5}{2})}r_\dA^4
 \]
 
 \item 
Furthermore, since 
$\nabla U(x)=0$ on $\Omega\backslash(\dA\cup\dB)$ and the diffusion is unchanged by a constant shift of the potential, we can assume that $U(x)=0$ on $\Omega\backslash(\dA\cup\dB)$. 
Since $G=\dA$ and since $U$ is continuous, $U(x)=0$ for all $x\in \partial G$. Hence $\sum_{i=1}^m e^{U(y_i)}=m$.
 
\item 
\label{lab:analytic}
Since $U$ is flat on $\tilde G \backslash G$, the first-passage function $h_{G,\tilde G^c}$ is that of a standard Brownian motion between two concentric spheres. The PDE in Equation (\ref{eq:pde}) of Appendix \ref{sec:three_perspectives} reduces to an instance of Laplace's equation, with analytic solution\cite{Wendel1980-sj}
 \[ h_{G, \tilde{G}^c} (x) = \frac{1}{r_\dA^{-3} - r_{\tilde{A}}^{ -3}} \| x
- x_A \|^{-3} - \frac{r_{\tilde{A}}^{-3}}{r_\dA^{-3} -
r_{\tilde{A}}^{-3}} \]
from which the gradient is found to be 
\[
\nabla h_{G, \tilde{G}^c} (x)  = \frac{-3}{r_\dA^{-3} - r_{\tilde{A}}^{-3}} \| x - x_A \|^{-4} \frac{x - x_A}{\| x - x_A \|} 
= \frac{-3}{r_\dA^{-3} - r_{\tilde{A}}^{-3}} \| x - x_A \|^{-4} \frac{x - x_A}{\| x - x_A \|} 
\]
And since $\| x - x_A \|=r_\dA$ and  $\frac{x - x_A}{\| x - x_A \|} 
=\normal(x)$,
\[
\normal(x)\cdot\nabla h_{G, \tilde{G}^c} (x)  = 
\frac{-3}
{r_\dA^4(r_\dA^{-3} - r_\tA^{-3})} 
\]

\item  It remains to choose a sample $y_1,\ldots,y_m$ from $\PMeasure$ on $G$, and estimates of the accompanying values of $h_{A,\tA^c}(y_i),\ i=1,\ldots,m$. Here, the shell method (Appendix \ref{sec:shell_method}) can be used, with $R=A$, $\tilde{R}=\tA$, and $S=G$, resulting in the desired samples $y_1,\ldots,y_m$ and estimates of $h_{A,\tA^c}(y_1),
\ldots,h_{A,\tA^c}(y_m)$,  say $u_1,\ldots,u_m$.

\end{enumerate}

Putting together the pieces, the approximation in (\ref{eqn:toy_cap_approx}) becomes
\begin{equation}
\label{eqn:approx_capacities}
\capac{A}{\tilde A} \approx
\frac{6\pi^{\frac{5}{2}}}
{\Gamma(\frac{5}{2})(r_\dA^{-3} - r_\tA^{-3})} 
\frac{1}{m}\sum_{i=1}^m u_i
\end{equation}
The approximation we used for $\capac{B}{\tB}$ is the same, but with the substitutions $A\to B$ and $\tA\to\tB$.

\subsection{Brownian Diffusion}
\label{sec:B_D}

The first set of experiments test the approach in the simplest possible case: there is no gradient in $U$ anywhere outside of the targets. Even though the resulting diffusion is just a standard Brownian motion, there are few symmetries in the configuration space and hence still no closed-form solution for the first-passage probabilities. There is however a large entropic barrier and an opportunity to dissect and test separately each component of the the overall approximation. 

\subsubsection{Are first-passage probabilities approximately constant outside of $\bm{\tA\cup\tB}$?}
\label{sec:toy_constant}
Is it true that in regions with golf-course potentials, the hitting probability is approximately constant over initial conditions that are modestly far away from the targets? In other words, is $h_{A,B}$  $\varepsilon$-flat on $\Omega\backslash(\tA\cup\tB)$?  Theorem \ref{thm:epsilon_flat} shows that this must hold in the limiting regime of small $r_A,r_B$ or large $n$, but it is not obvious whether the parameters of the current model lie in this regime.  

We ran 2,000 diffusion simulations at each of 100 randomly selected initial conditions in the region $\mathcal{B}(0, 1) \setminus (\tilde{A} \cup \tilde{B})$, yielding 100 well-estimated probabilities of hitting $A$ before $B$. A histogram of these probabilities can be be found in panel (a) of Figure \ref{fig:results}. The observations fall in a fairly narrow window, $[0.0820, 0.1185]$, suggesting that $h$ is indeed $\varepsilon$-flat, with $\varepsilon \approx 0.0365$. In addition, the empirical distribution is quite peaked.

\subsubsection{Are first-passage probabilities proportional to capacity?}
\label{sec:toy_capacity}
According to Theorem \ref{thm:main_thm},
in light of the $\varepsilon$-flatness of 
$h_{A,B}$  with respect to $\Omega\backslash(\tA\cup\tB)$, 
$h_{A,B}(x)$ should be well approximated by 
\begin{equation}
\label{eqn:capacity_ratio}
p_A \doteq \frac{\capac{A}{\tA}}{\capac{A}{\tA}+\capac{B}{\tB}}
\end{equation}
for $x\in 
\mathcal{B}(0, 1) \setminus (\tilde{A} \cup \tilde{B})$. 
In the current experiment, with $\nabla U=0$ on 
$\tA\backslash A$ and $\tB\backslash B$, the capacities can be computed analytically, e.g. from Corollary \ref{cor:flux}, 
with $S=A$ and $S=B$:
\begin{equation}
\label{eqn:analytic_capacities}
\capac{A}{\tA}  =
\frac{6\pi^{\frac{5}{2}}}
{\Gamma(\frac{5}{2})(r_A^{-3} - r_\tA^{-3})}  \qquad
\capac{B}{\tB}  =
\frac{6\pi^{\frac{5}{2}}}
{\Gamma(\frac{5}{2})(r_B^{-3} - r_\tB^{-3})}
\end{equation}
and hence 
\begin{equation*}
p_A = \frac{\frac{1}{r_A^{-3} - r_{\tilde{A}}^{-3}}}{\frac{1}{r_A^{-3} - r_{\tilde{A}}^{-3}} + \frac{1}{r_B^{-3} - r_{\tilde{B}}^{-3}}}
\approx 0.1100
\end{equation*}
This probability is within $2\%$ of the average found by direct simulations, which was $0.0975$. 

\subsubsection{Accuracy of the shell method.}
\label{sec:toy_shell}
Lastly, we check the accuracy of the capacity estimation algorithm that we call the shell method. Here, the exact values are available from the formulas in (\ref{eqn:analytic_capacities}):
$\capac{A}{\tA}=0.000637$ and $\capac{B}{\tB}=0.005151$. 
The approximate values come from two applications of (\ref{eqn:approx_capacities}), which produced $0.000591$ and $0.004827$, respectively.\footnote{With reference to Appendix \ref{sec:shell_method}, the following parameters were used to implement the shell method:
$m = 2, n = 4, N_p = 100, N_b = 3, N_s = 1000$, and a time-step of 
$10^{-7}$.}

Had we used the estimated capacities instead of the actual capacities for computing  $p_A$ from Equation
(\ref{eqn:capacity_ratio}) the approximation of the first-passage probability, given $x\in\mathcal{B}(0, 1) \setminus (\tilde{A} \cup \tilde{B})$, would have been 0.1091 instead of 0.1100.

In summary, all three components of the approach performed well in the Brownian motion test case.

\begin{figure}
\fbox{\begin{minipage}{\textwidth}
    \includegraphics[width=\textwidth]{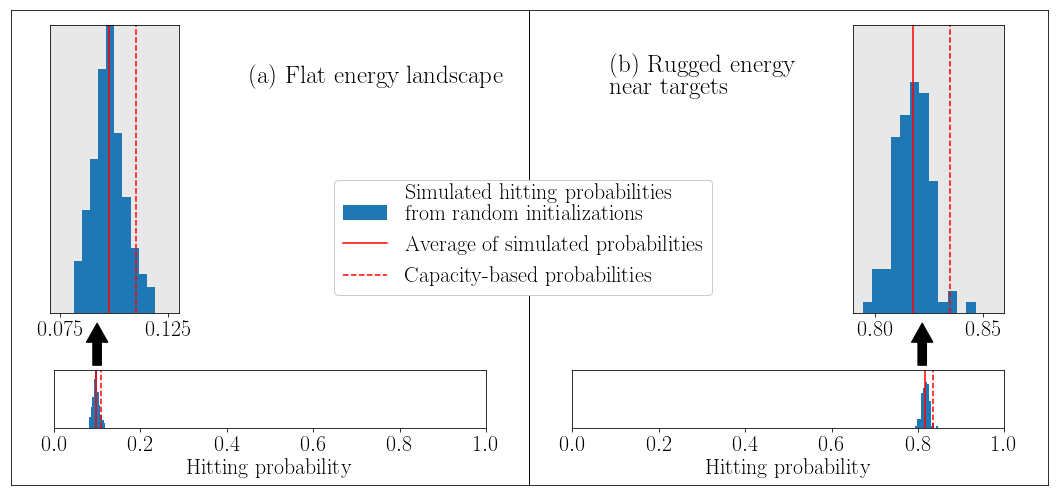}
    \caption{\label{fig:results} {\bf Hitting probabilities: direct simulations vs. evaluation of local capacities.}  Local capacities  can accurately answer the question ``where will we go next?''  We start by designating two targets, $A$ and $B$.  For any given initial condition, the task is to determine the probability that the diffusion will hit $A$ first. We study this question in the context of two different energy landscapes. In panel (a), we study the simplest possible case, where we have a flat energy landscape and the diffusion is simply a Brownian motion.  In panel (b), we consider the case that the energy landscape becomes complicated near the targets.  In both cases we consider 100 random initial locations which are modestly far  away from either target.  For each initial location, we conduct 2000 simulations to estimate the probability of hitting target $A$ first.  For both energy landscapes we see that the hitting probability is approximately constant for initial conditions which are modestly far away from the targets.  Furthermore, the value of this constant is very nearly proportional to the ratio of the local capacities. In summary, in these simulations first-passage probabilities are accurately estimated using only local simulations around the targets, and do not require any computations that involve the large space around $A$ and $B$.}
\end{minipage}}
\end{figure}

\subsection{Nontrivial Landscape}\label{sec:nontrivial_results}
We performed the same tests, but with a complex energy landscape in the neighborhoods of the targets, i.e. on $\dA\backslash A$ and $\dB\backslash B$. The details of the specification can be found in Appendix \ref{sec:energy_function}.

\subsubsection{Are first-passage probabilities approximately constant outside of $\bm{\tA\cup\tB}$?}
Following the same procedure used in \S\ref{sec:toy_constant}
we tested for the near-constancy of $h_{A,B}$ on 
$\Omega\backslash(\tA\cup\tB)$. The results are illustrated in Figure \ref{fig:results}, panel (b).  For each of the 100 initial conditions, the probability of first-passage at $A$ fell in the interval $[0.7985, 0.8460]$, consistent with the conclusions of Theorem \ref{thm:epsilon_flat}.

\subsubsection{Are first-passage probabilities proportional to capacity?}
In the previous experiments we were able to compute the capacities in closed form, and use $p_A$, from Equation (\ref{eqn:capacity_ratio}), to directly verify the conclusion of Theorem \ref{thm:main_thm}. This is not possible in the current experiments. We are forced instead to use the estimated capacities, from which we obtained the estimate 
$p_A\approx 0.8360$, which is within 2\% of $0.8175$,\footnote{Using the following parameters (Appendix \ref{sec:shell_method}): $m = 2$, $N_p = 3000$, $N_b = 5$, $N_s = 1000$, with time-step $10^{-6}$.  For target $A$ we used $n = 4$ and for the larger target $B$ we used $n = 5$.} the average found by direct simulations.

Notable in these results is the fact that the larger target, $B$, is substantially less likely than $A$ to be visited first. Evidently, the energy in $\dB\backslash B$, the region surrounding $B$, introduces a significant barrier to the diffusion process, at least in comparison to the energy surrounding $A$. Keep in mind that the process here is identical to the one in the previous experiments, i.e. Brownian motion, for so long as the process remains outside of 
$\dA\cup\dB$, and that in those experiments the first-passage occurred at $B$ approximately ten times more often than at $A$.
Evidently, the process is much more likely to exit $\dB\backslash B$ at $\partial\dB$ than at $\partial B$, at least in comparison to the dynamics in $\dA\backslash A$. These observations serve to further illustrate the role of local capacities and the importance of  their accurate approximation.

\subsubsection{Accuracy of the shell method.}

As already remarked, the nontrivial energy landscape precludes a direct assessment of the accuracy of the capacity estimation algorithm.  It is not possible to obtain exact values for the capacity.  Therefore, we cannot directly test whether our algorithm is accurately estimating the capacities.  However, the good agreement between the estimated value of $p_A$ and the results of straightforward simulation constitute indirect evidence supporting the approximations, and the accuracy of the shell method in particular.

Concerning computational efficiency, it is difficult to make a direct comparison between the capacity-estimation approach and straightforward simulation. There are many parameters and, besides, run times will depend on the dimension and details of the energies, possibly affecting the two approaches differently. In our experiments, for direct simulations we used the ``walk-on-spheres'' method to simulate trajectories in the flat region,\cite{bingham1972random} JIT compilation to remove loop overhead, multi-CPU parallelization, and the coarsest time step that yielded accurate results. As for capacity estimation, we made no effort to adjust the number of samples or the discretization parameters. Under these conditions, a single run of the direct simulation took about as long as estimating the two capacities.

Assume that  $h_{A,B}(x)$ is $\varepsilon$-flat relative to $\Omega \backslash (\tilde A \cup \tilde B)$. Both methods suffer if $\varepsilon$ is large: direct simulation because it will depend on the initial condition, which is unknowable in any realistic experiment, and capacity estimation because the capacity ratio has a built-in error (Theorem \ref{thm:main_thm}) that depends on $\varepsilon$. If we assume that $\varepsilon$ is negligible, then for a fixed $x\in\Omega \backslash (\tilde A \cup \tilde B)$ we can more-or-less directly compare the standard deviations of the capacity-based estimation of $p\doteq h_{A,B}(x)$ via the shell method to direct simulation via repeated-samples. For example, under the specific (albeit idealized) circumstances of the experiments in \S\ref{sec:B_D}, about how many direct-simulation samples would be needed to attain a confidence interval for $p$ of comparable width as from a single sample using capacity estimation by the shell method? The binomial estimator from $n$ direct samples has standard deviation $\sqrt\frac{p(1-p)}{n}$, which we can compare to the empirical standard deviation, $\hat{\sigma}\approx 0.006$, from 100 runs of the shell method estimating $p$. Recall from \S\ref{sec:B_D} that $p\approx 0.1100$, in which case approximately $n=2,700$ runs of direct simulation would be needed to get a comparable confidence interval. Bear in mind that a single run of the shell method requires about as much time as  a single run of direct diffusion. It would appear, then, that capacity-based estimation of first-passage probabilities can be several orders-of-magnitude faster than direct simulation.

\section{Discussion} 
\label{sec:Discussion}

Entropic barriers, which arise because of the difficulty of reaching a small number of target configurations out of a vast number of possible configurations, are in contrast to enthalpic barriers, which arise because of the difficulty of escaping local minima. The challenges of molecular dynamics result from the complicated energy landscapes associated with different biomolecular systems, and can largely be summarized by these two different kinds of barriers. Most of the existing methodology focuses on enthalpic barriers, where the central picture consists of an energy landscape with a multitude of local minima separated by high energetic barriers. A less studied picture includes entropic barriers---large 
flat regions of the energy landscape. Accurate simulations of  dynamical systems involving large single molecules, or multitudes of molecules, will typically require a better understanding of both kinds of barrier.

Sometimes, entropic barriers can be circumvented. We have provided conditions under which direct simulation of a diffusion across a nearly flat landscape, as in a ``golf-course type potential,'' can be 
replaced by the computation of the ``local capacity'' in the neighborhood of each target. The approach is appropriate when the first-passage probabilities, rather than first-passage times, are the objects of interest. Specifically, we give conditions 
under which the first-passage probabilities are approximately invariant to the initial configuration, as long as that configuration is moderately removed from each target region. In turn, a consequence of this invariance is that the hitting probabilities can be approximated using only local computations, or simulations, around the targets.  Numerical experiments on a prototypical entropic barrier, with a golf-course potential, demonstrate the validity of these results and the effectiveness of the approximations.

To what extent can these results contribute to the understanding of the folding of a large biopolymer? As a specific example, consider the folding of an RNA molecule, starting from either a denatured state or an intermediate state characterized by a non-native secondary structure. That is, some of the existing stems are not part of the native structure or additional stems will eventually appear. Assume that stem formation and destruction are sufficiently rapid to be considered as immediate events in the time-scale of folding. In the former case because stems, once seeded, are completed 
rapidly\cite{Porschke1977-xz}, and in the latter case because large-deviations, once they occur, occur rapidly. 
Given the current state, what happens next? 

Put aside, for the time being, the possibility of an existing stem unfolding, and assume that there are enough internal degrees of freedom that are sufficiently unconstrained so as to constitute an entropic barrier to the search for seeding a new stem. The available stems are easy to delineate, and each one can be considered a target, defined in terms of its own reaction coordinates. As shown here, target capacities can be estimated and used to construct a distribution over the ensemble of available stems. A choice from this distribution amounts to a shortcut around the entropic barrier. The geometry is certainly complicated, with boundaries defined by physical limits on bond configurations, including allowed dihedral angles, bond separations, and so-on. But movements that are within these constraints and at the same time distant from targets will likely be largely free of substantial gradients. In such cases the analyses presented here could be useful for predicting the next stem formation. 
Depending on the native structure, at some point in its folding trajectory the molecule may reach a state in which the secondary structure is sufficiently rich and constraining that the dynamics would face
little if any meaningful entropic barriers. In this regime, efforts to recapitulate the pathway would necessarily rely on direct simulation or other types of approximations, e.g. one of the variants of transition-path sampling.

More generally, these considerations will apply to the extent that a folding pathway can be viewed as 
a discrete-state random walk, from one secondary structure to a neighboring secondary structure (cf. Zhao et al.\cite{Zhao2010-zipping}),
and to the extent that each addition involves traversing large flat regions of configuration space. Obviously, there are many challenges, but perhaps chief among them is the loss of information about first-passage times. How are we to know whether a new stem will be seeded before an old stem unravels? More generally, what are the relative probabilities of transitions to neighboring states when some of these are defined by the unraveling of substructures? The answer may require a resolution, or marriage, of two approaches, one designed for enthalpic barriers, such as the barrier preventing a stem from unraveling, and the other designed for entropic barriers, such as the barrier in the way of the self intersection needed to seed a new stem. The former naturally addresses rates and time scales, but these are mostly unavailable in the approach to entropic barriers conceived here.

\newpage
\appendix
\noindent {\bf APPENDIX}

                                                         
\section{Three Perspectives on Hitting Probabilities}
\label{sec:three_perspectives}

Many of our results are based on the fact that hitting probabilities can actually be seen from three distinct perspectives.  Let $A,B\subset \Omega$, disjoint, open with smooth boundary.  

\begin{enumerate}
\item Hitting probabilities.  Let $\tau_S \triangleq \inf\{t:\ X_t \in S\}$ for any set $S$ and $h_{A,B}(x) \triangleq \mathbb{P}(X_{\tau_{A\cup B}}\in \partial A|X_0=x)$.
    
\item Elliptic equation.  Let $h^\mathrm{dir}_{A,B}(x)$ denote the solution to the partial differential equation:
	    \begin{align}\label{eq:pde}
    0 &= \sum_{i = 1}^n b_i (x) \frac{\partial u
        (x)}{\partial x_i} + \frac{1}{2} \sum_{i = 1}^n \sum_{j = 1}^n a_{ij} (x)
        \frac{\partial^2 u (x)}{\partial x_i \partial x_j}\quad x\notin \bar A,\bar B\\
    1 &= u(x),x\in  A \nonumber \\ 
    0 &= u(x),x\in  B \nonumber \\
    0 &= \normal(x)^Ta(x)\nabla u(x), x \in \partial{\Omega} 
    \nonumber
    \end{align}
This solution is unique and smooth.\cite{lieberman1986mixed}  What's more, it is equal to the hitting probability function: $h^\mathrm{dir}_{A,B}(x)=h_{A,B}(x)$ (cf.\ Section 6.7 of Chen\cite{chen2012symmetric}).
\item Variational form.  For any open set $S\subset \Omega$ let $\mathscr{E}_{S}(f,g)\triangleq \int_S \nabla f(x)^T a(x) \nabla g(x) \rho(dx)$ denote the ``Dirichlet Form'' of $f,g$ on the domain $S$.  Let $\mathscr L^2(S,\rho)$ denote the Hilbert space of functions on $S$ which are square-integrable with respect to $\rho$.  Let $\mathcal{H}^1(S,\rho)=W^{1,2}(S) \subset \mathscr{L}^2(S)$ denote the corresponding once-weakly-differentiable Hilbert Sobolev space.  We define $h^\mathrm{var}_{A,B}(x)$ as the solution to 
    \begin{align*}
    \min_{u \in \mathcal H^1(S)} \quad & \mathscr{E}_S(u,u) \\
    \mbox{subject to} \quad & u(x)=1,x\in \partial A \\
     & u(x)=0,x\in \partial B
    \end{align*}
    where $S=\Omega \backslash (A\cup B)$.  This solution is unique and equal to $h^\mathrm{dir}_{A,B}$ on $S$ (cf.\ Section 4 of Dret\cite{dret2016partial}).  
This variational perspective leads us to the notion of the ``condenser capacity'' associated with $h_{A,B}$.  It is defined as 
    \[
    \capac{A}{\Omega \backslash B} \triangleq \mathscr{E}_S(h_{A,B},h_{A,B})
    \]
    where again $S=\Omega \backslash (A\cup B)=(\Omega \backslash B) \backslash A$.  
\end{enumerate}
We will use all three of these perspectives to show our results.  For example, consider how the hitting probability perspective helps us show a result about capacities:
\begin{proposition}\label{prop:capacity}
Let $A\subset \tilde A,B\subset \tilde B$ with $\tilde A,\tilde B$ disjoint.  Then 
\[
\capac{A\cup B}{\tilde A \cup \tilde B}=\capac{A}{\tilde A}+\capac{B}{\tilde B}
\]
\end{proposition}
\begin{proof}
Since $\tilde A,\tilde B$ are disjoint and $X$ is continuous, the process cannot cross from one to the other without hitting the boundary.  Thus we have $\tau_{\partial \tilde A\cup \partial \tilde B \cup \partial A \cup \partial B}=\tau_{\partial \tilde A \cup \partial A}$ as long as $X_0\in\tilde A$.  We get a symmetric result if $X_0\in \tilde B$.  It follows that
\[
h_{A\cup B,(\tilde A\cup\tilde B)^c}(x) = 
    \begin{cases}
    h_{A,\tilde A^c}(x) & \mbox{if }x\in \tilde A\\
    h_{B,\tilde B^c}(x) & \mbox{if }x\in \tilde B\\
    \end{cases}
\]
We can now use this probabilistic perspective to help us understand the capacity by articulating it as the Dirichlet form on the relevant hitting probability functions
\begin{align*}
\capac{A\cup B}{\tilde A \cup \tilde B} 
        &= \int_{\tilde A\cup \tilde B \backslash (A\cup B)} \Vert \sigma \nabla h_{A\cup B,(\tilde A \cup \tilde B)^c}\Vert^2\rho(dx) \\
        &= \int_{\tilde A \backslash A} \Vert \sigma \nabla h_{A,\tilde A^c}\Vert^2\rho(dx)
            +\int_{\tilde B \backslash B} \Vert \sigma \nabla h_{B,\tilde B^c}\Vert^2 \rho(dx) \\
        &= \capac{A}{\tilde A}+\capac{B}{\tilde B}
\end{align*}
as desired.
\end{proof}


\section{Proof of Theorem \ref{thm:epsilon_flat}}
\label{sec:proof_epsilon_flat}

Consider the setup of the ``Toy Model'' as described in the Introduction.  A stationary reversible diffusion $X$ is trapped inside the unit $n$-dimensional ball.  We are interested to know which target $X$ will hit first: $A=\bb{x_A,r_A}$ or $B=\bb{x_B,r_B}$.  The function $h_{A,B}(x)$ indicates the probability we will hit $A$ first if $X_0=x$.  We will be particularly interested in the case where $X_0$ is outside of $\tilde A=\bb{x_A,r_{\tilde A}},\tilde B=\bb{x_A,r_{\tilde B}}$.  Also recall that in the toy example the diffusion behaves as a Brownian motion outside of $\dot A=\bb{x_A,r_{\dot A}},\dot B=\bb{x_B,r_{\dot B}}$.  

It turns out that by taking $n$ to be sufficiently high or $r_{\dot A},r_{\dot B}$ to be sufficiently small, we can make the hitting probabilities arbitrarily close to constant in the region away from $\tilde A,\tilde B$.  This is the content of Theorem \ref{thm:epsilon_flat} from the main text, which we restate here for the convenience of the reader:

\begingroup
\def\thetheorem{\ref{thm:epsilon_flat}}
\begin{theorem}  
If $\nabla U(x)=0$ on $x\in\Omega\backslash(\tA\cup\tB)$, then
for any fixed value of the dimension $n \geq 3$ and any $r_{\tilde{A}}, r_{\tilde{B}}, \varepsilon > 0$, there exists a constant $c=c(n, r_{\tilde{A}}, r_{\tilde{B}}, \varepsilon)$ such that if $r_{\dA}, r_{\dB} < c$ then 
$h_{A,B}(x)$ is 
$\varepsilon$-flat  relative to 
$\Omega / (\tilde{A} \cup \tilde{B})$.  
Likewise, for any fixed values of $r_{\dA}, r_{\tilde{A}}, r_{\dB}, r_{\tilde{B}}, \varepsilon>0$, there exists a constant $c=c(r_\dA, r_{\tilde{A}}, r_\dB, r_{\tilde{B}}, \varepsilon)$ such that if $n \geq c$ then 
$h_{A,B}(x)$ is
$\varepsilon$-flat relative to 
$\Omega / (\tilde{A} \cup \tilde{B})$. 
\end{theorem}
\addtocounter{theorem}{-1}
\endgroup

\begin{proof}
Let us assume $X_0 \notin \tilde A,\tilde B$.  There are essentially two things that must be proved:
\begin{enumerate}
\item The process $X$ converges to its stationary distribution fairly quickly.  This is supported by Lemma \ref{lem:uniform_ergodicity}.  Since $h_{A,B}(x) \in [0,1]$, this Lemma gives that
\[
|\mathbb{E}[h_{A,B}(M_t)-h_{A,B}(Z)]| \leq 2^2/4t =1/t
\]
where $Z$ is distributed according to the uniform distribution on $\Omega$ and $M$ is a Brownian motion trapped in $\Omega$ by normally reflecting boundaries.  

To connect this result on $M$ to our object of interest $h_{A,B}$, note that without loss of generality we may assume $X,M$ are on the same probability space and $M_t=X_t$ for $t<\tilde\tau \triangleq \inf \{t:\ X_t \in \partial \dot A \cup \partial \dot B\}$.  Moreover, note that $\forall t \geq 0$, $\tilde\tau\wedge t$ is a stopping time with finite expectation, and $h_{A,B}$ is a solution to the elliptic equation \ref{eq:pde}. Applying Dynkin's formula, we observe that 
\[
h_{A,B}(x) = \mathbb{E} [h_{A,B}(X_{t\wedge\tilde\tau})] = \mathbb{E} [h_{A,B}(M_{t\wedge\tilde\tau})]
\]
Thus
\begin{align*}
|h_{A,B}(x) - \mathbb{E}[h_{A,B}(Z)]| &= |\mathbb{E} [h_{A,B}(M_{t\wedge\tilde\tau})] - \mathbb{E}[h_{A,B}(Z)]|\\
&\leq |\mathbb{E} [h_{A,B}(M_{t\wedge\tilde\tau})] - h_{A,B}(M_t)]| + |\mathbb{E} [h_{A,B}(M_{t}) - h_{A,B}(Z)]|\\
&\leq \mathbb{P}(\tilde \tau \leq t) + 1/t
\end{align*}

\item The hitting time $\tilde \tau$ is generally long.  This is supported by Lemma \ref{lem:longtime}, which says that if $X_0 \not\in\tilde{A} \cup \tilde{B}$ we can ensure
\[
\mathbbm{P} \left( \tilde{\tau} \leq t \right) < \varepsilon/2
\]
for arbitrarily small $\varepsilon$ and arbitrarily large $t$, by taking $n$ sufficiently high or $r_{\dot A},r_{\dot B}$ sufficiently small.
\end{enumerate}
Putting these two results together at $t=2/\varepsilon$ we obtain that we can make $|h_{A,B}(x) - \mathbb{E}[h_{A,B}(Z)]| \leq \varepsilon / 2 + \varepsilon/2 = \varepsilon$ for arbitrarily small $\varepsilon$ by taking $n$ sufficiently high or $r_{\dot A},r_{\dot B}$ sufficiently small.
\end{proof}


\begin{lemma}
\label{lem:uniform_ergodicity}(Uniform ergodicity) Let $\Omega\subset\mathbbm{R}^d$ be a convex set with diameter $\xi$ and let $M$ denote a Brownian motion trapped by reflecting boundaries inside $\Omega$.  The distribution of $M_t$ converges uniformly to the uniform distribution, in the sense that:
\[ 
\sup_{f:\ \Omega\rightarrow [0,1]} | \mathbb{E}[f(M_t) - f(Z)|M_0=x]| \leq \xi^2/4t \qquad \forall x\in \Omega, t>0
\]
where $Z$ is uniformly distributed on $\Omega$.
\end{lemma}
\begin{proof}
Per Loper,\cite{Loper2018} 
\[
\sup_{f:\ \Omega\rightarrow [0,1]} | \mathbb{E}[f(M_t) - f(Z)|M_0=x]| \leq \mathbbm{P}(\tau>4t)
\]
where $\tau$ is the first exit time of a standard Brownian motion from $[-\xi,\xi]$. It is well known that $\mathbb{E}=\xi^2$. Hence, by the Markov inequality, $\mathbbm{P}(\tau>4t)\leq \xi^2/4t$.
\end{proof}

\begin{lemma}
\label{lem:longtime}Let $\tilde{\tau}= \inf \{ t : X_t \in \dot A \cup \dot B\}$. For any fixed value of $n
\geqslant 3$ and $r_{\tilde{A}}, r_{\tilde{B}}, \varepsilon,t > 0$, there
exists some $c (n, r_{\tilde{A}}, r_{\tilde{B}}, \varepsilon,t)$ such that if
$r_{\dot A}, r_{\dot B} < c$, then
\[ \mathbbm{P} \left( \tilde{\tau} < t |X_0 \not\in\tilde{A} \cup \tilde{B} \right) < \varepsilon \]
Likewise, for any fixed value of $r_{\dot A}, r_{\tilde{A}}, r_{\dot B}, r_{\tilde{B}},
\varepsilon,t > 0$, there exists some $c (r_{\dot A}, r_{\tilde{A}}, r_{\dot B},
r_{\tilde{B}}, \varepsilon,t)$, such that the same property holds as long as $n \geqslant c$.
\end{lemma}\begin{proof}
Our task is to show that $\tilde \tau$ is large with high probability.  Note that it suffices to show that $\tilde \tau_A \triangleq \inf\{t:\ X_t\in \dot A\}$ and $\tilde \tau_B \triangleq \inf\{t:\ X_t\in \dot B\}$ are both large with high probability and then apply a union bound.  Thus, without loss of generality we will focus on showing that $\tilde \tau_A$ is large with high probability.  For this, Chernoff's bound shows that it would be sufficient to show that we can make
\[
g(x)=\mathbb{E}[e^{-\frac{1}{2}\tilde \tau_A}|X_0=x]
\] 
arbitrarily small for every $x\notin \tilde A$. We note that $\forall X_0 \notin \tilde A$, the continuity of $X$ dictates that $X$ would have to cross the set $S_{\tilde r_{\dot A}}=\{x:\ |x-x_A| = \tilde r_{\dot A}\}$ before $\tilde \tau_A$, as long as $\tilde r_{\dot A} \in (r_{\dot A},r_{\tilde A})$.  Applying the strong Markov property of $X$, we obtain $g(x)\leq \sup_{y\in S_{\tilde r_{\dot A}}}g(y), \forall x \notin \tilde A$. As a result, to prove our theorem it will suffice to show that by taking $r_{\dot A}$ sufficiently small or $n$ sufficiently large we can ensure that $g(x)$ is uniformly arbitrarily small on $S_{\tilde r_{\dot A}}$ for some $\tilde r_{\dot A} \in (r_{\dot A},r_{\tilde A})$.  To show this, we will make use of two other hitting times:
\[
T =\inf \{t:\ X_t \notin \tilde A \backslash \dot A\} 
\qquad
T_1 =\inf \{t\geq T:\ X_t \in S_{\tilde r_{\dot A}}\} 
\]
$\forall X_0 \in S_{\tilde r_{\dot A}}$, if $X_T\notin \partial\dot A$, then the process exits $\tilde A \backslash \dot A$ at $\partial \tilde A$, and would have to cross $S_{\tilde r_{\dot A}}$ again before reaching $\dot A$. In this case, we have $T_1 \leq \tilde \tau_A$.  Applying the strong Markov property, we see that $\forall x\in S_{\tilde r_{\dot A}}$
\begin{align*}
	g(x) = &\mathbb{E}[e^{-\frac{1}{2}\tilde \tau_A}\indicatorf{X_T \in\partial\dot A}|X_0=x]+ \mathbb{E}[e^{-\frac{1}{2}\tilde \tau_A}\indicatorf{X_T \in\partial\tilde A}|X_0=x]\\
	    =  &\mathbb{E}[e^{-\frac{1}{2}\tilde \tau_A}\indicatorf{T=\tilde \tau_A}|X_0=x]+ \mathbb{E}[e^{-\frac{1}{2}\tilde \tau_A}\indicatorf{T\neq\tilde \tau_A}|X_0=x]\\
     = & \mathbb{E}[e^{-\frac{1}{2}\tilde \tau_A}\indicatorf{T=\tilde \tau_A}|X_0=x]+
     \mathbb{E}[e^{-\frac{1}{2}T_1}g(X_{T_1})\indicatorf{T\neq\tilde \tau_A}|X_0=x]\\
     \leq & \mathbb{E}[e^{-\frac{1}{2}T}\indicatorf{T=\tilde \tau_A}|X_0=x]+
                    \mathbb{E}[e^{-\frac{1}{2}T}\indicatorf{T\neq\tilde \tau_A}|X_0=x]\left(\sup_{y\in 
		   S_{\tilde r_{\dot A}}} g(y)\right)
\end{align*}
Note furthermore that the law of $e^{-\frac{1}{2}T},\indicatorf{T=\tilde \tau_A}$ is actually the same for every $x\in S_{\tilde r_{\dot A}}$, due to the fact that the diffusion behaves simply like a Brownian motion inside $\tilde A \backslash \dot A$ and the law of $T,\indicatorf{T=\tau_A}$ are thus functions of the one-dimensional diffusion of $|X_t-x_A|$.  In fact, Wendel\cite{Wendel1980-sj} gives explicit formulas for $\mathbb{E}[e^{-\frac{1}{2}T}\indicatorf{T=\tilde \tau_A}|X_0=x],\mathbb{E}[e^{-\frac{1}{2}T}\indicatorf{T\neq\tilde \tau_A}|X_0=x]$ which depend only upon $|x-x_A|,r_{\tilde{A}},r_{\dot A}$.  Applying this, taking the supremum over $x$ of our previous formula, and rearranging, we obtain
\begin{align*}
\sup_{x\in S_{\tilde r_{\dot A}}} g(x) &\leq \frac{\mathbb{E}[e^{-\frac{1}{2}T}\indicatorf{T=\tilde \tau_A}|X_0\in S_{\tilde r_{\dot A}}]}{1-\mathbb{E}[e^{-\frac{1}{2}T}\indicatorf{T\neq\tilde \tau_A}|X_0\in S_{\tilde r_{\dot A}}]} \triangleq L(r_{\dot A},\tilde r_{\dot A},r_{\tilde A},n)
\end{align*}
Applying Wendel's formulas, we obtain a closed form expression for $L$:
\[
=\frac{\left( \frac{r_{\dot A}}{\tilde r_{\dot A}} \right)^h (I_h (
r_{\tilde{A}}) K_h ( \tilde r_{\dot A}) - I_h ( \tilde r_{\dot A}) K_h (
r_{\tilde{A}}))}{
    \left( \frac{I_h ( r_{\tilde{A}})}{I_h ( \tilde r_{\dot A})} - \left(
\frac{r_{\tilde{A}}}{\tilde r_{\dot A}} \right)^h  \right)
K_h ( r_{\dot A}) I_h ( \tilde r_{\dot A})+ 
    \left( \left( \frac{r_{\tilde{A}}}{\tilde r_{\dot A}}
\right)^h  - \frac{K_h ( r_{\tilde{A}})}{K_h ( \tilde r_{\dot A})} \right) I_h (r_{\dot A})K_h ( \tilde r_{\dot A})}\\
\]
where $h=(n-2)/2$ and $I_h,K_h$ represent modified Bessel functions of the first and second kind of order $h$.  
 
Thus, to complete our proof, it suffices to show that we can drive $L(r_{\dot A}, \tilde r_{\dot A}, r_{\tilde A},n)$ to zero by taking $r_{\dot A}$ small or $n$ large:
\begin{itemize}
    \item When $r_{\dot A}$ is small.  The numerator of $L$ converges to zero as $r_{\dot A}\rightarrow 0$, because $\frac{r_{\dot A}}{\tilde r_{\dot A}}\rightarrow 0,h>0$ and the other terms are constant.  On the other hand, the denominator explodes, because as $r_{\dot A}\rightarrow0$ we have
    \begin{align*}
    & \left( \frac{I_h ( r_{\tilde{A}})}{I_h ( \tilde r_{\dot A})} - \left(\frac{r_{\tilde{A}}}{\tilde r_{\dot A}} \right)^h  \right)
                K_h ( r_{\dot A}) I_h ( \tilde r_{\dot A}) \rightarrow +\infty \\
    & \left( \left( \frac{r_{\tilde{A}}}{\tilde r_{\dot A}}\right)^h  - \frac{K_h ( r_{\tilde{A}})}{K_h ( \tilde r_{\dot A})} \right) 
                I_h (r_{\dot A})K_h ( \tilde r_{\dot A})\rightarrow 0
    \end{align*}
    These limits follow immediately from three properties of Bessel functions:
    \begin{itemize}
        \item $K_h(x)\rightarrow \infty,I_h(x)\rightarrow 0$ as $x\rightarrow 0$ for $h>0$
        \item $K_h(x),I_h(x)>0$ for $x>0,h>0$
        \item $\frac{I_{h}(y)}{I_{h}(x)} > \left(\frac{y}{x}\right)^{h}$ for $y>x$ and $h>0$
    \end{itemize} 
    The first two properties are well-known and can be found in DLMF\cite{noauthor_undated-ti}; the second can be found in Baricz.\cite{noauthor_undated-ti,baricz2010bounds}  In conclusion, since the numerator vanishes and the denominator explodes, we have that overall $L$ vanishes.  
    \item When $n$ is large.  It is clear that the previous result holds for any value of $\tilde r_{\dot A}\in (r_{\dot A}, r_{\tilde A})$.  For the large-$n$ case, we will be more picky: we will take $\tilde r_{\dot A}=\sqrt{r_{\dot A} r_{\tilde A}}$.  

    Let us look at the numerator first.  Asymptotics from the DLMF give that as $h=(n-2)/2 \rightarrow \infty$ we have
    \begin{align*}
    I_h(x) \sim \frac{x^h}{2^h\Gamma(h+1)} & & K_h(x) \sim \frac{2^h\Gamma(h+1)}{(2h)x^h}  
    \end{align*}
    Here by $f_1(h)\sim f_2(h)$ we mean ``asymptotic equivalence,'' i.e.\ $\lim_{h\rightarrow\infty}f_1(h)/f_2(h)=1$. Applying to our case:
    \begin{align*}
        \left(\frac{r_{\dot A}}{\tilde r_{\dot A}}\right)^h I_h (r_{\tilde{A}}) K_h ( \tilde r_{\dot A}) &
	\sim \frac{1}{2h}\left(\frac{r_{\tilde{A}} r_{\dot A}}{\tilde r_{\dot A}^2}\right)^h = \frac{1}{2h} \rightarrow 0 \\
        \left(\frac{r_{\dot A}}{\tilde r_{\dot A}}\right)^h I_h ( \tilde r_{\dot A}) K_h (r_{\tilde{A}}) &
	\sim \frac{1}{2h}\left(\frac{ r_{\dot A}}{r_{\tilde A}}\right)^h \rightarrow 0
    \end{align*}
    Putting these two limits together we see that the numerator of $L$ is asymptotically vanishing.

    Now let us turn to the denominator.  First note that 
    \begin{align*}
     \left( \frac{r_{\tilde{A}}}{\tilde r_{\dot A}}\right)^h  
	    I_h (r_{\dot A})K_h ( \tilde r_{\dot A}) &\sim \frac{1}{2h}\left(\frac{r_{\tilde{A}} r_{\dot A}}{\tilde r_{\dot A}^2}\right)^h = \frac{1}{2h} \rightarrow 0\\
     - \frac{K_h ( r_{\tilde{A}})}{K_h ( \tilde r_{\dot A})} 
	    I_h (r_{\dot A})K_h ( \tilde r_{\dot A}) &\sim -\frac{1}{2h}\left(\frac{r_{\dot A}}{r_{\tilde A}}\right)^h \rightarrow 0
    \end{align*}
    So those terms are negligible.  However, the other two terms of the denominator are in fact exploding: one to positive infinity and one to negative infinity.  To understand this delicate balance, we these we turn to Lemma \ref{lem:bessel}.  Applying this Lemma and the asymptotics of the DLMF, we obtain that
    \begin{align*}
    \left( \frac{I_h ( r_{\tilde{A}})}{I_h ( \tilde r_{\dot A})} - \left(\frac{r_{\tilde{A}}}{\tilde r_{\dot A}} \right)^h  \right)
                K_h ( r_{\dot A}) I_h ( \tilde r_{\dot A}) 
        &\geq \frac{r_{\tilde{A}}^h}{\cancel{I_h(\tilde r_{\dot A})}}\times\frac{r_{\tilde{A}}^{2} -\tilde r_{\dot A}^{2}}{2^{h+2}\Gamma(h+2)} K_h ( r_{\dot A}) \cancel{I_h ( \tilde r_{\dot A})} \\
        &\sim r_{\tilde{A}}^h\times\frac{r_{\tilde{A}}^{2} -\tilde r_{\dot A}^{2}}{2^{h+2}\Gamma(h+2)} \frac{2^h \Gamma(h+1)}{2hr_{\dot A}^h}  \\
        &=\left(\frac{r_{\tilde{A}}}{r_{\dot A}} \right)^h \left(\frac{r_{\tilde A}^2 - \tilde r_{\dot A}^2}{8h(h+1)}\right)
    \end{align*}
    which is indeed exploding as $h=(n-2)/2 \rightarrow \infty$.  

    In conclusion, with this choice of $\tilde r_{\dot A}$, the numerator is asymptotically vanishing, but the denominator is asymptotically exploding.  Thus $L$ vanishes.
\end{itemize}

Let us now review our arguments to see how the vanishing of $L$ leads to our desired conclusion.  We have just shown that by taking $n$ large or $r_{\dot A}$ small, we can ensure that $L$ vanishes.  Thus for any $x_0\notin \tilde A$ and any $\varepsilon'>0$, we can ensure
\begin{align*}
\varepsilon' & \geq L \geq \sup_{y \in S_{\tilde r_{\dot A}}} g(y) \geq g(x) \\
            & = \mathbb{E}[e^{-\frac{1}{2}\tilde \tau_A}|X_0=x]\\
            & \geq \mathbb{P}(\tilde \tau_A < t) e^{\frac{1}{2} t}
\end{align*}
In particular, for any $t,\varepsilon>0$ we can ensure that $\mathbb{P}(\tilde \tau_A < t)<\varepsilon/2$ by ensuring $\varepsilon' < \varepsilon e^{-\frac{1}{2} t}/2$.  We can then ensure that $\mathbb{P}(\tilde \tau < t)<\varepsilon$ by applying the same arguments to $\tilde \tau_B$ and applying a union bound.

\end{proof}

\begin{lemma}\label{lem:bessel}
If $a>b$ then
    \[
    \frac{I_h ( a)}{I_h (b)} - \left(\frac{a}{b} \right)^h 
    \geq
    \frac{a^h}{I_h(b)}\times\frac{a^{2} -b^{2}}{2^{h+2}\Gamma(h+2)}
    \]
\end{lemma}
\begin{proof}
Recall that $I_h$ may be defined as 
\[
I_h(x) = \sum_{m=0}^\infty \frac{x^{h+2m}}{2^{h+2m}\Gamma(m+h+1)\Gamma(m+1)}
\]
Thus
\begin{align*}
\frac{I_h ( a)}{I_h (b)} - \left(\frac{a}{b} \right)^h 
    &= \frac{I_h(a)b^h-a^hI_h(b)}{b^hI_h(b)}  \\
    &= \frac{a^hb^h\sum_{m=0}^\infty \frac{a^{2m} -b^{2m}}{2^{h+2m}\Gamma(m+h+1)\Gamma(m+1)}}{b^{h}I_h(b)}
\end{align*}
Since $a>b$, we have that $a^{2m} -b^{2m}$ is always positive.  Thus we can get a lower bound by simply taking one of the terms.  Choosing $m=1$, we get our final result.
\end{proof}


\section{Proof of Theorem \ref{thm:main_thm}}
\label{sec:proof_thm}

Let $A\subset\tilde A\subset\Omega,B\subset\tilde B\subset\Omega$.  Let $\tilde A,\tilde B$ be disjoint and
$h_{A,B}(x)$  $\varepsilon$-flat with respect to $\Omega \backslash (\tilde A \cup \tilde B)$. We assume the set boundaries are all smooth.  

Under these conditions, we will show we can use local capacities to get good approximations for $h_{A,B}(x)$ when $x\notin \tilde A,\tilde B$.  To do so, our key idea is to uncover upper and lower bounds on the value of the Dirichlet form applied to this function, $\mathscr{E}(h_{A,B},h_{A,B})$.  We will see that these bounds can be understood in terms of local capacities, and the resulting inequalities will then yield our main result in the form of Theorem \ref{thm:main_thm}.

\begin{lemma}  Let $S=\Omega \backslash (A\cup B)$.  The Dirichlet form of $h_{A,B}$ on $S$ can be upper-bounded in terms of the capacities:
\[ \mathscr{E}_S(h_{A,B}, h_{A,B}) \leqslant \frac{\tmop{cap} (A, \tilde{A}) \tmop{cap} (B,
\tilde{B})}{\tmop{cap} (A, \tilde{A}) + \tmop{cap} (B, \tilde{B})} \]
\end{lemma}
\begin{proof}
We recall from Appendix \ref{sec:three_perspectives} that
\[
\mathscr{E}_S(h_{A,B}, h_{A,B}) = \capac{A}{\Omega \backslash B} \leq \mathscr{E}_S (u,u)
\]
for any $u$ with $u(\partial A) = 1,u(\partial B) = 0$.  Thus, to prove an upper bound it suffices to find any such function for which we can calculate $\mathscr{E}(u,u)$.  To this end, consider
\[
u_c (x) \triangleq \left\{ \begin{array}{ll}
(1 - c) h_{A, \widetilde{A}^c} (x) + c & \tmop{if} x \in \tilde{A} \\
c (1 - h_{B, \widetilde{B}^c} (x)) & \tmop{if} x \in \tilde{B} \\
c & \mbox{otherwise}
\end{array} \right. 
\]
These functions are well-suited to giving us upper bounds on $\mathscr{E}_S (h_{A,B}, h_{A,B})$.  Indeed:
\begin{itemize}
\item $u_c(\partial A)=1,u_c(\partial B)=0$.  In fact, $u_c$ takes a constant value $c$ outside of $\tilde A,\tilde B$, drops smoothly in $\tilde B$ to achieve 0 on $\partial B$, and rises smoothly in $\tilde A$ to achieve 1 on $\partial \tilde A$.  
\item Noting that $u_c$ is written as a piecewise combination of hitting probability functions, we see that its Dirichlet form can be calculated in terms of capacities on local regions: $\mathscr{E}_S(u_c,u_c) = (1 - c)^2 \tmop{cap} (A, \tilde{A}) + c^2 \tmop{cap} (B, \tilde{B})$.
\end{itemize}
Thus the $u_c$ functions give us a practical way to calculate upper bounds: 
\[
\mathscr{E}_S(h_{A,B}, h_{A,B})\leq (1 - c)^2 \tmop{cap} (A, \tilde{A}) + c^2 \tmop{cap} (B, \tilde{B})
\]
This inequality holds for any value of $c$.  To get the best bound, we can take derivatives to minimize the right hand side with respect to $c$.  The result is 
\[
c^* = \frac{\capac{A}{\tilde A}}{\capac{A}{\tilde A}+\capac{B}{\tilde B}}
\]
Plugging this into the previous equation, we obtain our final result.
\end{proof}


\begin{lemma}  Let $S=\Omega \backslash (A\cup B)$.  Let $m = \frac{1}{2} (\sup_{x \notin \tilde A,\tilde B} h_{A,B} (x) + \inf_{x \notin \tilde A,\tilde B} (h_{A,B} (x)))$.  The Dirichlet form of $h_{A,B}$ can be lower-bounded in terms of $m$ and the capacities:
\[\mathscr{E}_S (h_{A,B}, h_{A,B}) \geq  
   \left( 1 - m - \frac{\varepsilon}{2} \right)^2 \tmop{cap} (A,\tilde{A}) \indicatorf{m\leq 1-\frac{\varepsilon}{2}} + 
   \left( m- \frac{\varepsilon}{2} \right)^2 \tmop{cap} (B, \tilde{B})\indicatorf{m\geq \frac{\varepsilon}{2}} 
\]
\end{lemma}
\begin{proof}
Recall that $\mathscr{E}_S(h_{A,B}, h_{A,B})$ can be expressed as an integral over $S$.  We decompose this into three integrals: one over $\tilde A$, one over $\Omega \backslash \tilde A,\tilde B$, and one over $\tilde B$.  
\[
\mathscr{E}_S (h_{A,B}, h_{A,B}) = \int_{\tilde A \backslash A} \Vert \sigma \nabla h_{A,B}\Vert^2 \rho(dx)
                                 + \int_{\tilde B \backslash B} \Vert \sigma \nabla h_{A,B}\Vert^2 \rho(dx)
                                 + \int_{\Omega \backslash \tilde A,\tilde B} \Vert \sigma \nabla h_{A,B}\Vert^2 \rho(dx)
\]
Since the integrand is always positive, we can get a lower bound by simply ignoring the integral over $\Omega \backslash \tilde A,\tilde B$ and focusing on the integrals over $\tilde A,\tilde B$.  The $\tilde A,\tilde B$ integrals can be lower-bounded using capacities.  

For example, let us focus on the $A$ integral.  There are two different possibilities we must consider:
\begin{itemize}
\item If $m>1-\varepsilon/2$ we will simply note that the integral over the $\tilde A$ region is non-negative.  
\item If $m\leq 1-\varepsilon/2$, then we define
    \[
    u_A (x) \triangleq \frac{h_{A,B} (x) - m - \frac{\varepsilon}{2}}{1 - m - \frac{\varepsilon}{2}}
    \]
    Note that  $h_{A,B}(x)=1$ for $x \in \partial A$ and the $\varepsilon$-flatness condition shows that $h_{A,B}(x) \leq m + \frac {\varepsilon}{2}$ for $x \in \partial \tilde A$.  Thus $u_A(\partial A)\geq1,u_A(\partial \tilde A)\leq 0$.  Lemma \ref{lem:inequalityboundaryvar} from Appendix \ref{sec:inequalityboundaryvar} may then be applied to yield that $\mathscr{E}_{\tilde A \backslash A}(u_A,u_A) \geq \capac{A}{\tilde A}$.  We can thus obtain the bound
    \begin{align*}
    \int_{\tilde A \backslash A} \Vert \sigma \nabla h_{A,B}\Vert^2 \rho(dx) 
	 &= \left(1 - m - \frac{\varepsilon}{2}\right)^2 \mathscr{E}_{\tilde{A} \backslash A}(u_A,u_A)\\
         &\geq \left(1 - m - \frac{\varepsilon}{2}\right)^2 \capac{A}{\tilde A}
    \end{align*}
\end{itemize}
Putting these two possibilities together, we obtain
\[
\int_{\tilde A \backslash A} \Vert \sigma \nabla h_{A,B}\Vert^2 \rho(dx) \geq \left(1 - m - \frac{\varepsilon}{2}\right)^2 \capac{A}{\tilde A}\indicatorf{m\leq 1-\frac{\varepsilon}{2}}
\]
Applying the same ideas to the integral over $\tilde B$, we obtain our result.
\end{proof}


We are now in a position to prove Theorem \ref{thm:main_thm} from the main text:

\begingroup
\def\thetheorem{\ref{thm:main_thm}}
\begin{theorem}
Assume that  $h_{A,B}(x)$ is $\varepsilon$-flat relative to 
$\Omega \backslash (\tilde A \cup \tilde B)$.
Then the first-passage probabilities can be well-approximated in terms of the target capacities:
\[ \sup_{x \notin \tilde A,\tilde B} \left| h_{A,B} (x) - \frac{\capac{A}{\tilde A}}{\capac{A}{\tilde A}+\capac{B}{\tilde B}} \right| \leqslant \varepsilon + \sqrt{\varepsilon/2} \]
\end{theorem}
\addtocounter{theorem}{-1}
\endgroup

\begin{proof}

To simplify notation, let $\capA=\capac{A}{\tilde A}$ and $\capB=\capac{B}{\tilde B}$.  Applying the previous two lemmas together, we obtain the inequality
\[
\frac{\capA \capB}{\capA +\capB}
\geq \mathscr{E}(h_{A,B},h_{A,B}) \geq
\left( 1 - m - \frac{\varepsilon}{2} \right)^2 \capA \indicatorf{m\leq 1-\frac{\varepsilon}{2}} + 
   \left( m- \frac{\varepsilon}{2} \right)^2 \capB \indicatorf{m\geq \frac{\varepsilon}{2}} 
\]
where $m = \frac{1}{2} (\sup_{x \notin \tilde A,\tilde B} h_{A,B} (x) + \inf_{x \notin \tilde A,\tilde B} (h_{A,B} (x)))$.  In analyzing this inequality, there are three possibilities to consider.  
\begin{itemize}
    \item If $m \in (\varepsilon/2,1-\varepsilon/2)$, the quadratic formula yields
    \begin{eqnarray*}
    m & \geqslant & \frac{\capA}{\capA + \capB} + \frac{
    \frac{\varepsilon}{2}(\capB - \capA) - 
    \sqrt{\capA \capB \varepsilon( 2 - \varepsilon)}}{\capA+\capB}\\
    m & \leqslant & \frac{\capA}{\capA + \capB} + \frac{
    \frac{\varepsilon}{2}(\capB - \capA) +
    \sqrt{\capA \capB \varepsilon( 2 - \varepsilon)}}{\capA+\capB}
    \end{eqnarray*}
    Applying $\left|\frac{\capB - \capA}{\capA + \capB}\right| \leq 1$ and the fact that the geometric mean $\sqrt{\capA \capB}$ never exceeds the arithmetic mean $(\capA+\capB) / 2$, it follows that
    \[
    \left|m -  \frac{\capA}{\capA + \capB}\right| \leq \frac{\varepsilon +\sqrt{\varepsilon (2-\varepsilon)}}{2}
    \]
    Applying the fact that $m$ was designed so that $|h_{A,B}(x) - m| <\varepsilon/2$ for all $x \notin \tilde A,\tilde B$, we obtain
    \[
    \left|h_{A,B}(x) -  \frac{\capA}{\capA + \capB}\right| \leq \frac{2\varepsilon +\sqrt{\varepsilon (2-\varepsilon)}}{2}
    \]

    \item If $m<\varepsilon/2$, our equations become 
    \[
    \frac{\cancel{\capA} \capB}{\capA +\capB}
    \geq
    \left( 1-m- \frac{\varepsilon}{2} \right)^2 \cancel{\capA} 
    \]
    Our assumption that $m \leq \varepsilon/2$ indicates that $(1- m- \varepsilon/2)^2 \geq (1-\varepsilon)^2$, thus in fact we have 
    \[
    \frac{\capB}{\capA +\capB}
    \geq
    \left( 1- \varepsilon \right)^2  = 1 + \varepsilon^2 - 2\varepsilon
    \]
    which means that $\capA/(\capA +\capB)
    \leq
    2\varepsilon - \varepsilon^2 \leq 2\varepsilon$.  
    Thus we assumed $m\in[0,\varepsilon/2]$ and showed that $\capA/(\capA+\capB)\in[0,2\varepsilon-\varepsilon^2]$, so it follows that 
    \[
    \left|m -  \frac{\capA}{\capA + \capB}\right| \leq 2\varepsilon -\varepsilon^2
    \]
    and so for any $x\notin \tilde A,\tilde B$, we have 
    \[
    \left|h_{A,B}(x) -  \frac{\capA}{\capA + \capB}\right| \leq 2.5\varepsilon -\varepsilon^2
    \]
    
    \item If $m>1-\varepsilon/2$, the same bound can be achieved by arguments which are symmetric to those employed in $m<\varepsilon/2$:
    \[
    \left|h_{A,B}(x) -  \frac{\capA}{\capA + \capB}\right| \leq 2.5\varepsilon -\varepsilon^2
    \]
\end{itemize}
Our final result is found by noting that all these bounds are upper-bounded by $\varepsilon + \sqrt{\varepsilon/2}$.
\end{proof}


\section{Proof of Proposition \ref{prop:flux}}
\label{sec:proof_proposition}

We first establish a lemma, using Green's first identity and some properties of the stationary SDE (\ref{equ:general_sde}), under the reversibility conditions (\ref{eqn:reversibility}), relative to $U$: 

\begin{lemma}  \label{lem:greenident}Fix some $S \subset \Omega$ with smooth boundary.  Then for any smooth function $g$ that satisfies $\mathcal{L}g = 0$ and smooth function $f$,
\begin{equation*}
\int_{S} \nabla f(x)^T a(x) \nabla g(x) e^{-U(x)}dx = \int_{\partial S} f(x) \normal(x)^T a(x) \nabla g(x) e^{-U(x)}\hausdorffmeasure
\end{equation*}
where $\normal$ are the normal vectors facing out of the set $S$ and $\hausdorffmeasure$ is the integral with respect to the $(n-1)$-dimensional Hausdorff measure and 
\[
(\mathcal{L}f)(x) \triangleq \sum_i b_i(x) \frac{\partial f (x)}{\partial x_i} + 
    \frac{1}{2} \sum_{ij}a_{ij}(x)\frac{\partial^2 f(x)}{\partial x_i \partial x_j} 
\]
\end{lemma}
\begin{proof}
First, we apply Green's first identity to get
\begin{align*}
	 &\int_{S} \nabla f^T a \nabla g e^{-U} dx\\
	=&\int_{\partial S}f \normal^T a \nabla g e^{-U} \hausdorffmeasure - \int_{S} f \nabla \cdot (a \nabla g e^{-U}) dx
\end{align*}
where
\begin{equation*}
	\nabla \cdot (a \nabla g e^{-U}) = \sum_{i}\frac{\partial}{\partial x_i}\left[\sum_{j}e^{-U}a_{i j}\frac{\partial g}{\partial x_j}\right]
\end{equation*}
Next, using the reversibility constraint on $b$ from Equation (\ref{eqn:reversibility}), it's not hard to verify that 
\begin{equation*}
	\nabla \cdot (a \nabla g e^{-U}) = 2e^{-U}\mathcal{L} g = 0
\end{equation*}
This gives us the desired result.
\end{proof}

\def\theproposition{\ref{prop:flux}}
\begin{proposition}
For any regions $G$ and $\tilde{G}$ having smooth boundaries and such that $A\subset G \subset \tilde G \subset \tilde A$, $\capac{A,\tA}$ can be expressed as a flux leaving $\tilde G \backslash G$:
\[
\ensuremath{\operatorname{cap}} (A, \tilde{A}) = \int_{\partial (\tilde G \backslash G)}  h_{A, \tilde{A}^c} (x)   \normal(x)^T a (x) \nabla h_{G, \tilde{G}^c} (x)e^{- U (x)} \hausdorffmeasure
\]
where $a(x)=\sigma(x)\sigma(x)^T$ is the diffusion matrix, $\hausdorffmeasure$ is the $(n-1)$-dimensional Hausdorff measure, and $\normal$ represents the outward-facing (relative to $\tilde G \backslash G$) normal vector on $\partial (\tilde G \backslash G)$.
\end{proposition}
\addtocounter{proposition}{-1}
\begin{proof}
First recall that $\mathcal{L} h_{A, \tilde{A}^c} = 0$ and
\begin{equation*}
	\tmop{cap}(A,\tilde A) = \int_{\tilde A \backslash A} \nabla h_{A, \tilde{A}^c} (x)^T a(x) \nabla h_{A, \tilde{A}^c} (x) e^{-U(x)} dx
\end{equation*}
Together with Lemma \ref{lem:greenident}, this yields that
\begin{equation}\label{eq:lemmaproof}
\tmop{cap}(A,\tilde A)  = \int_{\partial (\tilde A \backslash A)}  h_{A, \tilde{A}^c} \normal^T a  \nabla h_{A, \tilde{A}^c} e^{- U } \hausdorffmeasure = -\int_{\partial A}  \normal^T a  \nabla h_{A, \tilde{A}^c} e^{- U } \hausdorffmeasure
\end{equation}
where in the last step we used $h_{A,\tilde A^c}(x)=1,x\in \partial A,h_{A,\tilde A^c}(x)=0,x\in \partial \tilde A$.  Also, note that the normal vector on the right hand side is pointing {\em out of} the set $A$, as is our convention. Hence then negative sign.

Next we apply Lemma \ref{lem:greenident} again to get
\begin{align*}
	0 = \int_{G \backslash A} \cancel{\nabla (1)} a \nabla h_{A,\tilde A^c}^T e^{- U } dx = \int_{\partial (G\backslash A)} \normal^T a  \nabla h_{A,\tilde A^c} e^{- U } \hausdorffmeasure
\end{align*}
Combining this with Equation \ref{eq:lemmaproof} gives us
\begin{equation}
\label{eqn:cor}
\tmop{cap}(A,\tilde A)  = -\int_{\partial A}  \normal^T a  \nabla h_{A, \tilde{A}^c} e^{- U } \hausdorffmeasure = \int_{\partial G}  \normal^T a  \nabla h_{A, \tilde{A}^c} e^{- U } \hausdorffmeasure
\end{equation}
Using the facts that $h_{G,\tilde G^c}(x)=1,x\in \partial G,h_{G,\tilde G^c}(x)=0,x\in \partial \tilde G$, and $\mathcal{L} h_{G, \tilde G^c} = 0$, we apply Lemma \ref{lem:greenident} two more times to obtain 
\begin{align*}
	\tmop{cap}(A,\tilde A)  &= \int_{\partial G}  \normal^T a  \nabla h_{A, \tilde{A}^c} e^{- U } \hausdorffmeasure = \int_{\partial G}  h_{G, \tilde{G}^c} \normal^T a   \nabla h_{A, \tilde{A}^c} e^{- U } \hausdorffmeasure\\
				&= \int_{\partial (\tilde G \backslash G)}  h_{G, \tilde{G}^c} \normal^T a   \nabla h_{A, \tilde{A}^c} e^{- U } \hausdorffmeasure 
				= \int_{\tilde G \backslash G}  \nabla h_{G, \tilde{G}^c} a \nabla h_{A, \tilde{A}^c} e^{- U } dx\\
				&= \int_{\partial (\tilde G \backslash G)}  h_{A, \tilde{A}^c}    \normal^T a  \nabla h_{G, \tilde{G}^c} e^{- U } \hausdorffmeasure
\end{align*}
\end{proof}

\begin{corollary*}
For any region $S$ having smooth boundary $\partial S$, and such that $A\subset S \subset \tilde A$, $\capac{A,\tA}$ can be expressed as a flux leaving $S$:
\begin{equation*}
\ensuremath{\operatorname{cap}} (A, \tilde{A}) = \int_{\partial S}   \normal(x)^T a (x) \nabla h_{A, \tilde{A}^c} (x)e^{- U (x)} \hausdorffmeasure
\end{equation*}
where $a$ and $\hausdorffmeasure$ are as defined in the Proposition, and  $\normal$ is the outward-facing normal on $\partial S$.
\end{corollary*}
\begin{proof}
Put $G=S$ in Equation (\ref{eqn:cor}).
\end{proof}

\section{Inequality Boundary Conditions for the Variational Form}
\label{sec:inequalityboundaryvar}

Recall that $h_{A,B}(x)$ may be defined variationally.  We have let 
\[
\mathscr{E}(f,g)\triangleq \int_\Omega \nabla f(x)^T a(x) \nabla g(x) \rho(dx)
\]
denote the ``Dirichlet Form.''  Let $\Omega \subset \mathbb{R}^n$ compact and open with smooth boundary.  Let $\mathscr L^2(\bar \Omega,\rho)$ denote the Hilbert space of functions on $\bar \Omega$ which are square-integrable with respect to a continuous positive measure $\rho(dx)=e^{-U}dx$.  Let $\mathcal{H}^1(\bar \Omega,\rho)=W^{1,2}(\bar \Omega,\rho) \subset \mathscr{L}^2(\bar \Omega,\rho)$ denote the corresponding once-weakly-differentiable Hilbert Sobolev space.  Let $A,B\subset \Omega$, open, disjoint, with smooth boundary, and define $\capac{A}{\Omega \backslash B} \in \mathbb{R}$ as the minimizing value of the problem
    \begin{align*}
    \min_{u \in \mathcal H^1} \quad & \mathscr{E}(u,u) \\
    \mbox{subject to} \quad & u(x)=1,x\in A \\
     & u(x)=0,x\in B
    \end{align*}

It is natural to consider an apparently different problem, where the equality boundary conditions are replaced with inequalities.  Here we show that it is not possible to get lower than $\capac{A}{\Omega \backslash B}$ by such a relaxation.

\begin{lemma} \label{lem:inequalityboundaryvar} Let $\tilde h$ satisfy $\tilde h(x)\geq 1$ on $A$ and $\tilde h(x)\leq 0$ on $B$.  Then $\mathscr{E}(\tilde h,\tilde h) \geq \capac{A}{\Omega \backslash B}$.
\end{lemma}
\begin{proof}
Let $k=\mathtt{clamp}(\tilde h,0,1)$, i.e.
\[
k(x)=
\begin{cases}
\tilde h(x) & \tilde h(x)\in[0,1] \\
0 & \tilde h(x)\leq 0 \\
1 & \tilde h(x)\geq 1 \\
\end{cases}
\]  
Note that $k\in \mathcal{H}^1$ and satisfies the equality boundary conditions.  Thus, by definition, $\capac{A}{\Omega \backslash B} \leq \mathscr{E}(k,k)$.  This immediately yields our result:
\begin{align*}
\capac{A}{\Omega \backslash B} \leq \mathscr{E}(k,k) &= \int \Vert \sigma \nabla k \Vert^2 \rho(dx) \\
    &=\int_{x:\ \tilde h(x)\in[0,1]} \Vert \sigma \nabla \tilde h \Vert^2 \rho(dx) \\
    &\leq \int \Vert \sigma \nabla \tilde h \Vert^2 \rho(dx) = \mathscr{E}(\tilde h,\tilde h)
\end{align*}

\end{proof}

\section{Shell Method}
\label{sec:shell_method}

The algorithm for estimating local hitting probabilities is outlined as follows:

{\algorithm{\label{alg:hitting_prob_estimation}Estimating $h_{R,\tilde R^c}(x)$ for many values of $x$ on a shell $\partial S$

\begin{description}
    \item[Input] $R \subset S \subset \tilde{R} \subset \Omega$ and a stationary reversible  diffusion process $\{X_t\}_{t \geq 0}$ in $\Omega$ with invariant measure ${\mu}= e^{- U (x)} \mathrm{d} x$. We also require a series of subsets
  \[ R = S_0 \subset S_1 \subset \cdots \subset S_{m - 1}
     \subset S_m = S \subset S_{m + 1} \subset \cdots \subset S_n = \tilde{R} \]
     which indicate a kind of reaction coordinate.

    \item[Output] A collection of points $z_1,\cdots z_{N_p}$ on $\partial S$ sampled from the invariant measure ${\mu}= e^{- U (x)} \mathrm{d} x$ restricted on $\partial S$, along with estimates of $h_{R,\tilde R^c}(z_i)$ for each point.

\end{description}
     
\begin{enumerate}
\item Discretize the space.  
\begin{enumerate}
  \item Generate an ensemble of samples $z_1, \ldots, z_{N_p}$ on $\partial S$
  according to the invariant measure ${\mu}= e^{- U (x)} \mathrm{d} x$ restricted to $\partial S$. 
  
  \item Evolve the ensemble on $\partial S$, by repeatedly sampling an initial location from the uniform distribution on $\{ z_1, \ldots, z_{N_p} \}$ and carry out a local simulation following the dynamics of $\{X_t\}_{t \geq 0}$
  until the trajectory hits either $\partial S_{m - 1}$ or $\partial S_{m +
  1}$. Record the hitting locations on $\partial S_{m - 1}$ and $\partial S_{m + 1}$ until we have $N_p$ points on both $\partial S_{m - 1}$ and $\partial
  S_{m + 1}$. In most cases, the process is more likely to hit one of $S_{m - 1}, S_{m + 1}$ than the other, and we need to run more than $2N_p$ local simulations to get at least $N_p$ samples on both shells. We store the results of the redundant local simulations for future estimation of transition probabilities.
  
  \item Repeat the above process to sequentially evolve the ensembles on $\partial S_{m - 1}, \ldots,
  \partial S_2$ and on $\partial S_{m + 1}, \ldots, \partial S_{n - 2}$, to get $N_p$ samples on all of the intermediate shells $\partial S_1,
  \ldots, \partial S_{n - 1}$. Store the results of the redundant local simulations for future
  estimation of transition probabilities.
  
  \item For each one of the shells $\partial S_1, \ldots, \partial S_{n -
  1}$, cluster the $N_p$ samples on that shell into $N_b \,$ states. In our implementation, we use $k$-means, and represent the $N_b$ states by the $N_b$ centroids we get from the algorithm.
\end{enumerate}

The result of this step is a partitioning of each shell $\partial S_i$ into
$N_b$ regions, representing an adaptive discretization of the shells. For a point on a shell $\partial S_i$, we identify its corresponding discrete state by finding the nearest centroid.

\item Estimate the transition probabilities between these discrete states by running an additional $N_s$ local simulations for each one of the $N_b$ states on each shell. The result of this step is an estimate of the probability of transitioning from state $k$ on $\partial S_i$ to state $l$ on $\partial S_j$, which we denote by $P^{(i, j)}_{k, l}$, where $k, l \in \{ 1, \ldots, N_b \}$ and $i, j \in \{ 1, \ldots, n - 1 \}$ with $| i - j | = 1$.
  
\item Use the transition probabilities to get an estimate of the hitting probabilities for the $N_b$ states on
  $\partial S$.  In line with related works on Markov state models,\cite{Pande2010-yi, Chodera2014-bh, Husic2018-xp} we  approximate the continuous dynamics using closed-form calculations from the discrete Markov chain we have developed in the previous two steps.  In particular, we estimate overall hitting probabilities using the standard ``one-step analysis.'' For any $k \in \{ 1, \ldots, N_b \}$ and $i \in \{ 1, \ldots, n - 1 \}$, let $u^{(i)}_k$ denote the probability of hitting $\partial R = \partial S_0$ before hitting $\partial \tilde{R} = \partial S_n$ if we start the discretized process at state $k$ on $\partial S_i$.  We can calculate our object of interest by solving the matrix difference equation \[ u^{(i)} = P^{(i, i + 1)} u^{(i + 1)} + P^{(i, i - 1)} u^{(i - 1)}, i = 1, \ldots, n - 1 \] with boundary conditions $u^{(0)} = {\mathbf{1}}, u^{(n)} = {\mathbf{0} }$, where ${\mathbf{0}}$ and ${\mathbf{1}}$ are vectors of all $0$'s and $1$'s. This gives the estimated hitting probability for each discrete state.  We then estimate the hitting probability of each point $z_i$ by
  \begin{equation}
      \label{equ:hitprobest}
  h_{R,\tilde R^c}(z_i) = u^{(m)}_k, z_i \in \text{state }k\text{ on }\partial S\end{equation}
\end{enumerate}
}}

\section{Details on Energy Function} \label{sec:energy_function}

For the energy function, we hand-designed two different kinds of landscape: random well energy, which we use for the region around target $A$, and random crater energy, which we use for the region around target $B$. The basic components of these energy functions are the well component, given by
\begin{equation}
F_w(x|d_w, r) = -\frac{d_w}{r^4}(||x - x_A||_2^4 - 2r^2||x - x_A||_2^2) - d_w
\end{equation}
where $d_w$ gives the depth of the well; the crater component, given by
\begin{equation}
F_c(x|d_c, h, r) = \frac{d_c}{3b^2r^4 - r^6}(2||x - x_B||_2^2 - 3(b^2 + r^2)||x - x_B||_2^4 + 6b^2r^2||x - x_B||_2^2) - d_c
\end{equation}
where $d_c$ and $h$ give the depth and the height of the crater, respectively, and
\begin{equation}
b^2 = -\frac{1}{3d}(-3 d_c r^2 + C + \frac{\Delta_0}{C})
\end{equation}
with
\begin{align*}
C &= 3r^2 \sqrt[3]{d_c h (d_c + \sqrt{d_c (d_c + h)})} \\
\Delta_0 &= -9 d_c h r^4
\end{align*}
and finally a random component, given by
\begin{equation}
F_r(x|\mu, \sigma) = \sum_{i=1}^m\prod_{j=1}^d exp(-\frac{(x_j - \mu_{i j})^2}{2\sigma_{i j}^2})
\end{equation}
where $\mu=(\mu_{i j})_{m \times d}$ and $\sigma=(\sigma_{i j})_{m \times d}$, with $\mu_i=(\mu_{i 1}, \cdots, \mu_{i, d}), i=1, \cdots, m$ being the locations of $m$ Gaussian random bumps in the region around the targets, and $\sigma_{i j}, i=1, \cdots, m, j=1, \cdots, d$ gives the corresponding standard deviations.

To make sure the energy function is continuous, and the different components of the energy function are balanced, we introduce a mollifier, given by
\begin{equation}
F_m(x|x_0, r) = exp(-\frac{r}{r - ||x - x_0||_2^{20}})
\end{equation}
where $x_0=x_A, r=r_{\dot A}$ or $x_0=x_B, r=r_{\dot B}$, depending on which target we are working with, and a rescaling of the random component, which is given by $0.1 * d_w$ if we are perturbing the well component, and $0.1 * (d_c + h)$ if we are perturbing the crater component.

Intuitively, for the well component, we use a $4$th order polynomial to get a well-like energy landscape around the target that is continuous and differentiable at the boundary. Similarly, for the crater component, we use a $6$th order polynomial to get a crater-like energy landscape around the target that is also continuous and differentiable at the boundary. For the random component, we are essentially placing a number of Gaussian bumps around the target. And for the mollifier, we are designing the function such that it's almost exactly 1 around the target, until it comes to the outer boundary, when it transitions smoothly and swiftly to 0.
To summarize, given parameters $d_w, d_c, h$ and random bumps $\mu_{A}, \mu_{B}$ with $\mu^{A}_i\in \dot{A} \setminus A, i=1, \cdots, m_A, \mu^{B}_i \in \dot{B} \setminus B, i=1, \cdots, m_B$, and the corresponding standard deviations $\sigma^{A}, \sigma^{B}$ with $\sigma^{A}_{i j}, i=1,\cdots, m_A, j=1, \cdots, d, \sigma^{B}_{i j}, i=1, \cdots, m_B, j=1, \cdots, d$, the energy function we used in the experiments is given by
\begin{align}
F(x) &= F_w(x|d_w, r_{\dot A}) + 0.1 \times d_w \times F_m(x|x_A, r_{\dot A}) + F_r(x|\mu^{A}, \sigma^{A}), \forall x \in \dot{A} \setminus A \\
F(x) &= F_c(x|d_c, h, r_{\dot B}) + 0.1 \times (d_c + h) \times F_m(x|x_B, r_{\dot B}) + F_r(x|\mu^{B}, \sigma^{B}), \forall x \in \dot{B} \setminus B
\end{align}

In our actual experiments, we used
\begin{equation*}
d_w = 10.0, 
d_c = 6.0,
h = 1.0,
\sigma^{A}_{i j} = \sigma^{B}_{k, l} = 0.01, \forall i, j, k, l
\end{equation*} and
\begin{equation*}
\mu^{A} = \begin{pmatrix}%
0.512&0.583&-0.013&0.013&-0.001\\%
0.464&0.575&-0.001&0.019&-0.014\\%
0.503&0.611&-0.012&-0.024&0.023\\%
0.5&0.601&-0.024&0.034&0.011\\%
0.486&0.586&0.006&0.01&0.001\\%
0.489&0.588&-0.017&0.002&0.027\\%
0.493&0.585&0.015&-0.001&-0.032\\%
0.516&0.596&0.027&-0.026&0.022\\%
0.514&0.624&0.01&0.01&-0.002\\%
0.5&0.605&0.017&-0.016&0.004%
\end{pmatrix}, 
\mu^{B} = \begin{pmatrix}%
-0.696&-0.006&0.023&-0.041&0.019\\%
-0.731&0.021&-0.033&-0.014&0.017\\%
-0.694&-0.034&-0.009&0.031&0.019\\%
-0.666&-0.013&0.002&0.017&0.009\\%
-0.68&0.058&0.007&-0.011&-0.008\\%
-0.704&-0.022&0.034&0.003&0.026\\%
-0.714&-0.015&0.017&0.027&0.028\\%
-0.681&0.017&-0.046&-0.04&-0.002\\%
-0.648&-0.009&0.002&-0.012&-0.022\\%
-0.664&-0.04&0.05&-0.012&-0.002%
\end{pmatrix}
\end{equation*}

\bibliographystyle{plain}
\bibliography{refs.bib}

\end{document}